\documentclass[10pt]{article}
\usepackage{amsfonts, latexsym, amssymb, amsthm, amsmath, mathrsfs, esint, color, authblk, fullpage}
\usepackage{amscd}
\usepackage{enumerate}
\usepackage{pstricks}
\usepackage{pst-node}
\usepackage{mathtools}
\usepackage{graphicx}
\graphicspath{  }
\usepackage{relsize}
\usepackage{hyperref}
\usepackage{comment}
\usepackage{scrlfile}
\usepackage{srcltx}

\renewcommand{\a}{\alpha}

\def\e{\varepsilon}

\newcommand{\R}{\mathbf{R}}

\def\mr{\mathbf{R}}

\def\o{\omega}
\def\l{\lambda}

\renewcommand{\O}{\Omega}

\DeclareMathOperator{\Sp}{Sp}

\newcommand\N{\mathbf{N}}
\newcommand\Z{\mathbf{Z}}

\newcommand\supp{\mathop{\operatorname{supp}}}
\newcommand\dist{\operatorname{dist}}

\DeclareMathOperator{\diam}{diam}

\def\a{\alpha}

\def\XXint#1#2#3{{\setbox0=\hbox{$#1{#2#3}{\int}$}
     \vcenter{\hbox{$#2#3$}}\kern-.5\wd0}}

\newcounter{bei}

\renewcommand{\div}{\mathrm{div}\,}

\newtheorem{theorem}{Theorem}[section]
\newtheorem{lemma}[theorem]{Lemma}
\newtheorem{proposition}[theorem]{Proposition}
\newtheorem{corollary}[theorem]{Corollary}
\newtheorem{assumption}[theorem]{Assumption}
\newtheorem{definition}[theorem]{Definition}
\newtheorem{remark}[theorem]{Remark}
\theoremstyle{remark}

\usepackage{mathtools} 
\newcommand{\randomspace}{\Omega}
\newcommand{\pspace}{S}
\newcommand{\randomelement}{\omega}

\newcommand{\randommeasure}{P}
\newcommand{\td}{\mathrm{d}}
\newcommand{\drandommeasure}{\mathrm{d}\randommeasure(\randomelement)}

\newcommand{\wtwoscale}{\xrightharpoonup{2}}

\newcommand{\stwoscale}{\xrightarrow{2}}
\newcommand{\weakly}{\rightharpoonup}

\newcommand{\stochsobolev}{\mathcal W}
\newcommand{\stochsmooth}{\mathcal C^\infty}

\newcommand{\eps}{\varepsilon}

\DeclarePairedDelimiter\abs{\lvert}{\rvert}
\DeclarePairedDelimiter\norm{\lVert}{\rVert}
\DeclarePairedDelimiter\scalar{\langle}{\rangle}
\DeclarePairedDelimiter\set{\{}{\}}

\newcommand\restrict[2]{{
  \left.\kern-\nulldelimiterspace 
  #1 
  \vphantom{\big|} 
  \right|_{#2} 
  }}

\usepackage{color}
\usepackage[color]{changebar}
\cbcolor{red}
\begin{document}

\title{\sc Stochastic homogenisation of high-contrast media}

\author[1]{Mikhail Cherdantsev}
\author[2]{Kirill Cherednichenko}
\author[3]{Igor Vel\v{c}i\'{c}\,}
\affil[1]{School of Mathematics, Cardiff University, Senghennydd Road, Cardiff, CF24 4AG, United Kingdom}
\affil[2]{Department of Mathematical Sciences, University of Bath, Claverton Down, Bath, BA2 7AY, United Kingdom}
\affil[3]{Faculty of Electrical Engineering and Computing, University of Zagreb, Unska 3, 10000 Zagreb, Croatia}
\maketitle

\vspace{-6mm}
\begin{center}
  \textsl{To the memory of Vasilii Vasil'evich Zhikov}
\end{center}

\begin{abstract}
 
Using a suitable stochastic version of the compactness argument of [V. V. Zhikov, 2000. On an extension of the method of two-scale  convergence and its applications. {\it Sb. Math.,} {\bf 191}(7--8), 973--1014], we develop a probabilistic framework for the analysis of heterogeneous media with high contrast. We show that an appropriately defined multiscale limit of the field in the original medium satisfies a system of equations corresponding to the coupled ``macroscopic" and  ``microscopic'' components of the field, giving rise to an analogue of the ``Zhikov function'', which represents the effective dispersion of the medium. We demonstrate that, under some lenient conditions within the new framework, the spectra of the original problems converge to the spectrum of their  homogenisation limit.

 \vskip 0.5cm



\end{abstract}

\section{Introduction}

Asymptotic analysis of differential equations with rapidly oscillating coefficients has featured prominently among the interests of the applied analysis community during the last half a century. The problem of understanding and quantifying the overall behaviour of heterogeneous media has emerged as a natural step within the general progress of material science, wave propagation and  mathematical physics. In this period several frameworks have been developed for the analysis of families of differential operators, functionals and  random processes describing multiscale media, all of which have benefitted from the invariably deep insight and mathematical elegance of the work of V. V. Zhikov. In the present paper we touch upon two subjects in which his contributions have inspired generations of followers:
 the stochastic approach to homogenisation, in particular through his collaboration with S. M. Kozlov during the 1980s, and the analysis of differential operators describing periodic composites with high contrast, which started with his fundamental contribution \cite{Zhikov2000}.

Our present interest in the context of stochastic homogenisation of high-contrast composites stems from the relationships that have recently been indicated 
between media with negative material properties (``metamaterials''), and more generally time-dispersive media, and ``degenerate'' families of differential operators, where {\it e.g.} loss of uniform ellipticity of the symbol is known to lead to non-classical dispersion relations in the limit of vanishing ratio $\varepsilon$ of the microscopic ($l$) and macroscopic ($L$) lengths: 
$\varepsilon=l/L\to0.$ 
The work \cite{Zhikov2000} has provided an example, in the periodic context, of what one should expect in the limit as $\varepsilon\to0$ in terms of the two-scale structure of the solution as well as the spectrum of the related differential operator, in the case when the metamaterial is modelled by disjoint ``soft'' inclusions with low, order $O(\varepsilon^2)$ values of the material parameters (say, elastic constants in the context of linearised elasticity), embedded in a connected  ``stiff'' material with material constants of order $O(1).$ In mathematical terms, the coefficients of the corresponding differential expression alternate between values of different orders in $\varepsilon,$ where the contrast increases as $\varepsilon$ gets smaller. 

In the present article we introduce a stochastic framework for the analysis of homogenisation problems with soft inclusions and explore the question on what version of the results of \cite{Zhikov2000} can be achieved in this new framework. In particular, we are interested in the equations that describe the stochastic two-scale limit, in an appropriate sense, of the sequence of solutions to the probabilistic version of a Dirichlet problem in a bounded domain of ${\mathbf R}^n.$  Furthermore, we show that the spectra of such problems converge, in the Hausdorff sense, to the spectrum of the limit problem, which we analyse in a setting that models distributions of soft inclusions whose shapes are taken from a certain finite set and whose sizes vary over an interval.  To our knowledge, the present manuscript is the first work containing an analysis of random heterogeneous media with high contrast that results in a ``complete'' Hausdorff-type convergence statement for the spectra of the corresponding differential operators. Various aspects of multiscale analysis of high-contrast media in the stochastic context have been looked at in a handful of papers, {\it e.g.}      
\cite{BMP}, \cite{BBM}, \cite{Bellieud}. 


While in the periodic context norm-resolvent convergence results been obtained for high-contrast media, see \cite{ChC, ChChC} , the stochastic case remains open to developments of a similar nature. It is anticipated that the operator-theoretic approach to problems of the kind we discuss in the present article will provide a general description of the types of spectral behaviour that can occur in the real-world applications where it is difficult to enforce periodicity of the microstructure. On the other hand, as we show in the present work, new wave phenomena should be expected in the stochastic setting ({\it e.g.} a non-trivial continuous component of the spectral measure of the homogenised operator for a bounded-domain problem), which makes the related future developments even more exciting.

Next, we outline the structure of the paper. In Section \ref{stochastic} we recall the notion of stochastic two-scale convergence, which we use, in Sections \ref{problem_form}, \ref{lim_eq}, to pass to the limit, as $\varepsilon\to0,$ in a family of homogenisation problems with random soft inclusions. In Section \ref{problem_form} we give a formulation of the high-contrast problem we study and provide some auxiliary statements. In Section \ref{lim_eq} we describe the limit problem and prove the strong resolvent convergence of the $\varepsilon$-dependent family to the limit system of equations. In Section \ref{inclusions_completeness} we provide a link between the spectra of the Laplacian operator on realisations of the inclusions and of the corresponding stochastic Laplacian.  In Section \ref{spectral_convergence} we prove that sequences of normalised eigenfunctions of the $\varepsilon$-dependent problems are compact in the sense of strong stochastic two-scale convergence. Finally, in Section \ref{examples_sec} we discuss two examples of the general stochastic setting and describe the structure of the corresponding limit spectrum.

\


 In conclusion of this section we introduce some notation used throughout the paper. For a Banach space $X$ and its dual $X^*,$ we denote by $_X\langle \cdot, \cdot\rangle_{X^*}$ the corresponding duality. For a Hilbert space $H$ the inner product  of $a,b \in H$ is also denoted by $\scalar{a,b}_H$ and, if $H=\mathbf{R}^n,$ by $a \cdot b$. For a set $\mathcal{O}$ we denote by $\chi_\mathcal{O}$ its characteristic function, which takes value one on the set $\mathcal{O}$ and zero on the complement to $\mathcal{O}$ in the appropriate ambient space. For $D\subseteq \mathbf{R}^n$ we denote by $\overline{D}$ its closure and by $|D|$ its Lebesgue measure. Further, we use the notation $B_r(0)$ for the ball in $\mathbf{R}^n$ of radius $r$ with the centre at the origin; $Y$ denotes the cube $[0,1)^n$ with torus topology, where the opposite faces are identified; 
and 
$\mathbf{N}_0^l:=\{0,\dots,l\}$. 
 For an operator $\mathcal{A}$ on some Hilbert space, we denote by $\Sp \mathcal{A}$ its spectrum. Finally, for a Lipschitz open set $D\subset{\mathbf R}^n,$ we denote by $-\Delta_D$ the (positive) Laplace operator with the Dirichlet boundary condition on $\partial D.$ For $x \in \mathbf{R}^n$ we denote by $[x]$ the element of $\mathbf{Z}^n$ which satisfies $[x]\leq x<[x]+(1,\dots,1)$. For $k=1,\dots,n$ by $e_k$ we denote the $k$-th coordinate vectors.
 
\section{Stochastic two-scale convergence} 
\label{stochastic}
\subsection{Probability framework} 
Let $(\Omega, \mathcal{F},P)$ be a complete probability space. We assume that $\mathcal{F}$ is countably generated, which implies that the spaces $L^p(\Omega)$, $p \in [1,\infty),$ are separable. For a function $u \in L^1(\Omega),$ we will sometimes write $\langle u \rangle$ for $\int_{\Omega} u$. 
By $S$ we denote a domain (open bounded subset) in $\mathbf{R}^n$. 
\begin{definition}\label{defgroup}
	A family $(T_x)_{x \in \mathbf{R}^n}$ of measurable bijective mappings $T_x:\Omega \to \Omega$ on a probability space $(\Omega, \mathcal{F}, P)$ is called a dynamical system on $\Omega$ with respect to $P$ if:
	\begin{enumerate}
		\item $T_x \circ T_y=T_{x+y}\ $ $\forall\,x,y \in \mathbf{R}^n$;
		\item $P(T_{x} F )=P(F)\ $ $\forall x \in \mathbf{R}^n$, $F \in \mathcal{F}$;
		\item $\mathcal{T}: \mathbf{R}^n \times\Omega \to \Omega,\ $ $(x, \omega)\to T_x (\omega)$ is measurable (for the standard $\sigma$-algebra on the product space, where on $\mathbf{R}^n$ we take the Lebesgue $\sigma$-algebra).
	\end{enumerate}
\end{definition}
We next define the notion of ergodicity for dynamical systems introduced above.
\begin{definition}\label{defergodic}
A dynamical system is called ergodic if one of the following equivalent conditions is fulfilled:
\begin{enumerate}
	\item $f$ measurable, $f(\omega)=f(T_x \omega) \ \forall x \in \mathbf{R}^n, \textrm{ a.e. } \omega \in \Omega \implies f(\omega)\textrm{\ is\ constant\ $P$-a.e.\ } \omega \in \Omega$.
        \item\label{def:ergodicityb} $P\bigl((T_x B \cup B) \backslash (T_x B \cap B)\bigr)=0\ \forall x \in \mathbf{R}^n \implies P(B) \in \{0,1\}$.  
\end{enumerate}	
\end{definition} 
Henceforth we assume that the dynamical system $T_x$ is ergodic.

\begin{remark}\label{rem:ergo-weak}
Note that for the condition \ref{def:ergodicityb} the implication $P(B) \in \set { 0,1}$ has to hold, if the symmetric difference between $T_x B$ and $B$ is a null set.
It can be shown (see, e.g., \cite{cornfeld}) that ergodicity is also equivalent if a priori only the weaker implication
\[
T_x B = B\ \ \forall x \in \mr^n \implies P(B) \in \set{ 0,1}
\]
holds. 
\end{remark}

For $f \in L^p(\Omega),$ we write $f(x, \omega):=f(T_x \omega)$, defining the realisation $f \in L^p_{\rm loc}(\mathbf{R}^n, L^p(\Omega))$. 
There is a natural unitary action on $L^2(\Omega)$ associated with $T_x:$
\begin{equation} \label{matthaus1000}
 U(x)f=f \circ T_x, \quad \quad f \in L^2(\Omega).
 \end{equation} 
It  can be shown that the conditions 
of Definition \ref{defgroup} imply that this is a strongly continuous group (see \cite{zhikov2}).
It is often necessary that the set of full measure be
invariant in the sense that together with the point $\omega$ it contains the whole
"trajectory" $\{T_x\omega, x \in \mathbf{R}^m\}$. This requirement can always be met on the
basis of the following simple lemma (see \cite[Lemma 7.1]{zhikov2}).
\begin{lemma} \label{marko100} 
 Let $\Omega_0$ be a set of full measure in $\Omega$. Then there exists a set of full
measure $\Omega_1$ such that $\Omega_1 \subseteq \Omega_0$, and for a given $\omega \in \Omega_1$ we have $T_x\omega \in \Omega_0$
for almost all $x \in \mathbf{R}^m.$
\end{lemma} 	

For each $j=1,2,...n,$ we define 
the infinitesimal generator $D_j$ of the unitary group $U_{x_j}$ by the formula 
\begin{equation}
  D_j f(\omega)=\lim_{x_j \to 0}  \frac{f(T_{x_j}\omega)-f(\omega) }{x_j},\quad\quad   f \in L^2(\Omega),
  \label{generator}
\end{equation}
where the limit is taken in $L^2(\Omega).$  Notice that $i D_1,$ $j=1,\dots,n,$
are commuting, self-adjoint, closed and densely defined linear operators on the separable Hilbert space $L^2 (\Omega)$. 
The domain $\mathcal{D}_i(\Omega)$ of such an operator is given by the set of $L^2(\Omega)$-functions for which the limit (\ref{generator}) exists. We consider 
the set
\begin{equation}
   W^{1,2}(\Omega):=\bigcap_{j=1}^n\mathcal{D}_j(\Omega) 
   \label{W12}
\end{equation}
and similarly
\begin{align*}
	W^{k,2} (\Omega)&:=\bigl\{f \in L^2(\Omega): D_1^{\alpha_1}\dots D_n^{\alpha_n} f \in L^2(\Omega),\; \alpha_1+\cdots +\alpha_n=k\bigr\},\\[0.5em]
	W^{\infty,2} (\Omega)&:= \bigcap_{k \in \mathbf{N}} W^{k,2}(\Omega).
\end{align*}
It is shown by the standard semigroup property that $W^{\infty,2}(\Omega)$ is dense in $L^2(\Omega)$. 
We also define the space
\[
   \mathcal{C}^{\infty} (\Omega)= \set[\big] {f \in W^{\infty,2} (\Omega): \forall (\alpha_1,\dots, \alpha_n) \in \mathbf{N}_0^n \quad D_1^{\alpha_1} \dots D_n^{\alpha_n} f \in L^{\infty} (\Omega) }. 
\]
By the smoothening procedure discussed in \cite[p.\,232]{zhikov2} (see also the text before Lemma \ref{lemmagloria} below), it is shown that $\mathcal{C}^{\infty} (\Omega)$ is dense in $L^p(\Omega)$ for all $p \in [1,\infty)$ as well as in $W^{k,2}(\Omega)$ for all $k.$ Furthermore, it is shown that $W^{1,2} (\Omega)$ is separable. 
Notice that, due to the infinitesimal generator being closed, $D_i f$ can be equivalently defined as the function that satisfies the property 
\begin{equation} \label{matthaus1}
   \int_{\Omega} D_i f\,g =-\int_{\Omega } f\,D_i g\quad \forall g \in \mathcal{C}^{\infty}(\Omega).
\end{equation}
If $f \in W^{1,2} (\Omega),$ we may also define 
$D_if (x,\omega):=D_if(T_x \omega)$ for all $x\in\mathbf{R}^n.$ It can be shown that the following identity holds (see \cite{gloria1}):
\begin{equation} \label{matthaus0} 
\begin{aligned}
W^{1,2} (\Omega) &=\bigl\{ f \in W^{1,2}_{\rm loc}\bigl(\mathbf{R}^n, L^2(\Omega)\bigr): f(x+y,\omega)= f(x, T_y \omega) \quad \forall x,y,\ \text{a.e.\ }  \omega \bigr\}
\\[0.4em]
&=\bigl \{ f \in C^1\bigl(\mathbf{R}^n, L^2(\Omega)\bigr): f(x+y,\omega)= f(x, T_y \omega) \quad  \forall x,y,\ \text{a.e.\ }  \omega\bigr\}. 
\end{aligned}\end{equation}
Moreover, for a.e. $\omega \in  \Omega$ the function $D_if (\cdot,\omega)$ is the distributional derivative of $f(\cdot,\omega ) \in L^2_{\rm loc} (\mathbf{R}^n):$ a proof of this fact can be found in \cite[Lemma~A.7]{gloria1}.

Following \cite{sango1}, we denote by $\norm \cdot_{\#,k,2}$ the seminorm on $\mathcal{C}^{\infty} (\Omega)$ given by 
\[
  \norm u_{\#,k, 2}^2=\sum_{\alpha \in \mathbf N^n,\ \abs { \alpha } = k} \norm {D^\alpha u}^2_{L^2(\Omega)}.
\]
By $\mathcal{W}^{k,2}(\Omega)$ we denote the completion of $\mathcal{C}^{\infty}(\Omega)$ with respect to the seminorm $\norm \cdot_{\#,k,2}$.
The gradient operator $\nabla_{\omega}:=(D_{1},\dots, D_n)$ and the operator $\div_{\omega}:= \nabla_{\omega} \cdot$ are extended uniquely by continuity to mappings from $\mathcal{W}^{1,2}(\Omega)$ to $L^2(\Omega,\mathbf{R}^n)$ and from $\mathcal{W}^{1,2}(\Omega,\mathbf{R}^n)$ to $L^2(\Omega),$ respectively. Finally, by a density argument, it is easily seen that $\mathcal{W}^{1,2}(\Omega)$ is also the completion of $W^{1,2}(\Omega)$ with respect to $\norm \cdot_{\#,1,2}$.

\subsection{Definition and basic properties}

The key property of ergodic systems is the following theorem, due to Birkhoff (for a more general approach, see \cite{Ackoglu}).
\begin{theorem}[``Ergodic Theorem"] \label{thmergodic}
	Let  $(\Omega,\mathcal{F}, P)$ be a probability space with an ergodic dynamical system $(T_x)_{x \in \mathbf{R}^n}$ on $\Omega$. Let $f \in L^1(\Omega)$ be a function and $\mathcal{B}\subset \mathbf{R}^n$ be a bounded open set. Then for $P$-a.e. $\omega \in \Omega$ one has
        \begin{equation}\begin{aligned}\label{eq:birkhoff}
        \lim_{\varepsilon \to 0} \int_{\mathcal{B}}f(T_{x/\varepsilon}\omega) \td x= |\mathcal{B}| \int_{\Omega} f(\omega) \drandommeasure.
        \end{aligned}\end{equation}
	Furthermore, for all $f \in L^p(\Omega)$, $1 \leq p \leq \infty,$ and a.e. $\omega \in \Omega$, the function $f (x, \omega)= f(T_x \omega)$ satisfies $f(\cdot, \omega) \in L^p_{\rm loc} (\mathbf{R}^n)$. For $p < \infty$ one has $f(\cdot/ \varepsilon, \omega)=f(T_{ \cdot/ \varepsilon} \omega ) \rightharpoonup \int_{\Omega} f \td P$ weakly in $L^p_{\rm loc}(\mathbf{R}^n )$ as $\varepsilon \to 0$. 
	
\end{theorem}
The elements $\omega$ such that \eqref{eq:birkhoff} holds for every $f \in L^1(\Omega)$ and bounded open $\mathcal{B}\subset{\mathbf R}^n$ are refereed to as typical elements, while the corresponding sets $(T_x \omega)_{x \in \mr^N}$ are  called typical  trajectories. Note that the separability of $L^1(\Omega)$ implies that almost every $\omega \in \Omega$ is typical, and in what follows we only work with such $\omega.$

 For vector spaces $V_1, V_2,$ we denote by by $V_1 \otimes V_2$ their usual tensor product. 
 We define the following notion of stochastic two-scale convergence, which is a slight variation of the definition given in \cite{zhikov1}. We shall stay in the Hilbert setting ($p=2$), as it suffices for our analysis.

\begin{definition}\label{definicija1}
Let $(T_{x}\randomelement)_{x \in \R^n}$ be a typical trajectory and $(u^\eps)$ a bounded sequence
in $L^2(S)$. We say that $(u^\eps)$ {weakly stochastically two-scale converges} to 
$u \in L^2(\pspace \times \randomspace )$ 
and write
  $u^\eps \wtwoscale u,$ if 
\begin{equation}
  \lim_{\eps \downarrow 0} \int_{\pspace} u^\eps(x)g\bigl(x,T_{\eps^{-1} x} \randomelement\bigr) \td x
  = \int_{\randomspace} \int_\pspace u(x,\randomelement) g(x,\randomelement) \td x\, \drandommeasure\quad\quad \forall g \in C^\infty_0(\pspace) \otimes \stochsmooth(\randomspace).
  \label{conv_def}
\end{equation}

If additionally 
 $ \norm  { u^\eps}_{L^2(S)} \to \norm { u}_{L^2(S \times \Omega)},$
we say that $(u^\eps)$ strongly stochastically two-scale converges to $u$ and write $u^\eps \stwoscale u$.

\end{definition}

\begin{remark} 
The convergence of $(u^{\varepsilon})$ is defined along a fixed typical trajectory and {\it a priori} the limit depends on this trajectory. In applications, such as the analysis of the PDE family in Section \ref{lim_eq}, it often turns out that the limit does not depend on the trajectory chosen. For this reason, and to simplify notation, in what follows we often do not indicate this dependence explicitly.

Note also that, by density, the set of admissible test functions $g$ in (\ref{conv_def}) can be extended to $ L^2(S) \otimes L^2(\Omega)$. 
\end{remark}

In the next proposition we collect the properties of stochastic two-scale convergence that we use in the present work. 

\begin{proposition} \label{properties}  The following properties of stochastic two-scale convergence hold.
	\begin{enumerate}
	\item
Let $(u^\eps)$ be a bounded sequence in $L^2(S)$.
Then there exists a subsequence (not relabeled) and $u \in L^2(S\times\Omega)$ 
such that $u^{\eps} \wtwoscale u$.
\item If $u^{\e} \wtwoscale u$ then 
$ \|u\|_{L^2(S \times \Omega)} \leq \liminf_{\e \to 0} \|u^{\e}\|_{L^2 (S)}.$
\item If $(u^\eps) \subseteq L^2(S)$ is a bounded sequence with $u^\eps \to u$ in $L^2(S)$ for some $u \in L^2(S)$, then $u^\eps \stwoscale u$.
\item If $(v^\eps) \subseteq L^{\infty} (S)$ is uniformly bounded by a constant and $v^\eps \to v$ strongly in $L^1(S)$ for some $v \in L^\infty(S)$, and  $(u^\eps )$ is bounded in $L^2(S)$ with  $u^\eps \wtwoscale u$ for some $u \in L^2(S \times \Omega),$ then $v^\eps u^\eps \wtwoscale vu$.  
\item Let $(u^\eps)$ be a bounded sequence in $W^{1,2}(S)$. Then on a subsequence (not relabeled) 
$u^\eps \weakly u^0$ in $W^{1,2}(S),$ and there exists $u^1 \in L^2( S, \stochsobolev^{1,2}(\randomspace))$ 
such that 
\begin{align*}	
\nabla u^\eps \wtwoscale \nabla u^0 + \nabla_\omega u^1(\cdot, \omega)\,.
\end{align*}
\item Let  $(u^\eps)$ be a bounded sequence in $L^2(S)$ such that $\varepsilon \nabla u^\eps$ is bounded in $L^2(S, \mathbf{R}^n)$. Then there exists $u \in L^2( S, W^{1,2}(\randomspace))$ such that on a subsequence 
\begin{equation}
	u^{\eps}  \wtwoscale u,\quad\quad 
	\eps \nabla u^{\eps}  \wtwoscale \nabla_{\omega} u(\cdot,\omega). 
\end{equation}
\end{enumerate} 
\end{proposition}

\begin{proof}
	In view of analogies with the periodic case, we just give a sketch of the proof. 
	A proof of (a) can be found in \cite[Lemma~5.1]{zhikov1}. 
	 For the proof of (b), we take an arbitrary $g \in  C^\infty_0(\pspace) \otimes \stochsmooth(\randomspace)$ and calculate
	 \begin{equation*}
	 \begin{aligned}
	  \liminf_{\e \to 0} \int_S\bigl|u^\e(x)&-g(x, T_{\e^{-1}x}\omega)\bigr|^2 dx\\[0.5em]
	  &=\liminf_{\e \to 0} \biggl(\int\bigl|u^{\e}(x)\bigr|^2dx
	  -2\int_S u^{\e}(x)g(x, T_{\e^{-1}x}\omega)dx
	  +\int_S\bigl|g(x, T_{\e^{-1}x}\omega)\bigr|^2dx\biggr) \\[0.5em] & = \liminf_{\e \to 0} \biggl(\int\bigl|u^{\e}(x)\bigr|^2dx
	  -2\int_{S \times \Omega}  u(x,\omega)g(x,\omega)dx dP(\omega)
	  +\int_{S \times \Omega}\bigl|g(x,\omega)\bigr|^2dx dP(\omega)\biggr).
	  \end{aligned}
	   \end{equation*}
	   We obtain the claim by density and letting $g \to u$.
	 The proof of (c), (d) is straightforward. The proof of (e) goes in the same way as in the periodic case, by the duality argument.
	 First, one proves that if $f \in (L^2 (\Omega))^n$ is such that
	 \[
	 \int_{\Omega} f\cdot g=0 \quad\quad \forall g \in\bigl\{g\in\mathcal{C}^{\infty} (\Omega,\mr^n): 
	 {\rm div}_{\omega}\,g=0\bigr\},
	 \]
	 and therefore there exists $\psi \in \mathcal{W}^{1,2} (\Omega)$ such that 
	 $f=\nabla_{\omega} \psi,$ and one then proceeds in the same way as in the periodic case (see \cite{Allaire}). In order to show the claim  (f), take the subsequence such that $u^{\eps} \wtwoscale u$, where $u \in L^2(S \times \Omega)$ and $\e \nabla u^{\e} \wtwoscale z$, where $z \in L^2(S \times \Omega, \mathbf{R}^n)$. We choose the test functions of the form $\varphi^{\e}(x)=a(x) b(T_{\e^{-1}x}\omega)$, where $a \in C_0^{\infty}(S)$ and $b \in \mathcal{C}^{\infty} (\Omega)$ and by partial integration we conclude 
	 \begin{equation*}
	 	\lim_{\e  \to 0}\int_S \e \nabla u^{\e}(x) \varphi^{\e} (x)
		=-\int_S \int_{\Omega} u(x, \omega)a(x) \nabla_{\omega} b(\omega) dx dP(\omega)=
		\int_S \int_{\Omega} z(x) a(x)b(\omega) dx dP(\omega),
	 \end{equation*}
	 from which the claim follows by a density argument, in view of the property \eqref{matthaus1}. 
\end{proof} 
 
 \section{Problem formulation and auxiliary statements} 
 \label{problem_form}
 Let $S \subseteq \mathbf{R}^n$ be a bounded domain.  We take $\mathcal{O}\subseteq \Omega$ such that $0<P(\mathcal{O})<1$ and for each $\omega\in\Omega$ consider
 its ``realisation"
 $$ \mathcal{O}_\omega=\{x \in \mathbf{R}^n: T_x \omega \in\mathcal{O}\}. $$
We assume that the following conditions are satisfied. 
\begin{assumption} \label{kirill100}
For a.e. $\omega \in \Omega$ one has 
\begin{equation}
\mathcal{O}_{\omega}:=\bigcup_{k=1}^\infty \mathcal{O}_{\omega}^k,
\label{components}
\end{equation}
where: 
 
1) $\mathcal{O}_{\omega}^k,$ $k\in{\mathbf N},$ are  open connected sets with Lipschitz boundary;

2) For a.e. $\omega \in \Omega$ one has $\overline{\mathcal{O}_{\omega}}=\cup_{k=1}^\infty \overline{\mathcal{O}_{\omega}^k};$ 

3) There exist $c_1,c_2>0$ such that $c_1 \leq \diam \mathcal{O}_{\omega}^k \leq c_2\  \forall k \in \mathbf{N};$


4) There exists a sequence of disjoint bounded domains $B^k_{\omega}$ such that $\overline{\mathcal{O}^k_{\omega}} \subseteq B^k_{\omega},$ $k \in \mathbf{N},$ and $C_{\omega}>0$ such that for all $k \in \mathbf{N}$ the following extension property holds: for all $u \in W^{1,2}(B^k_{\omega}\backslash \overline{\mathcal{O}^k_{\omega}})$ there exists $\widetilde u \in  W^{1,2}(B^k_{\omega})$ satisfying
\begin{equation*}
\widetilde u=u \textrm{ on } B^k_{\omega} \backslash \overline{\mathcal{O}^k_{\omega}},\quad\quad\quad
\int_{B^k_{\omega}} |\nabla \widetilde u|^2 \leq C_{\omega} \int_{B^k_{\omega}\backslash \mathcal{O}^k_{\omega}}  |\nabla u|^2, \quad\quad\quad \Delta \widetilde u=0 \textrm{ on } \mathcal{O}^k_{\omega}. 
\end{equation*}
 \end{assumption}
 It is easily seen that Assumption \ref{kirill100} holds for the examples given in Section \ref{misha20}. Denote by $\Lambda$ the set of typical elements $\omega\in\Omega$ satisfying the conditions listed in Assumption \ref{kirill100}, and for all $\omega\in\Lambda,$ $\varepsilon>0$ define $S_0^\e(\omega)$ as the union of all components 
 $\e \mathcal{O}^k_{\omega}$ that are subsets of $S$ and stay sufficiently far from its boundary, in the sense that there exists $C=C(\omega)>0$ such that 
 \begin{equation}
 S_0^\e(\omega):=\bigcup_{k \in K^{\e}_{\omega}} \e \mathcal{O}^k_{\omega},\quad\quad K^\varepsilon_{\omega}:=\bigl\{ k \in \mathbf{N}: \, \e \mathcal{O}_\o^k \subseteq S,\ \dist\bigl(\e \mathcal{O}_\o^k, \partial S\bigr)>C\e\bigr\}.
\label{K_def} 
\end{equation}
 We denote the complement of the set $S_0^{\e}(\omega)$ by $S_1^\e(\omega) : = S\setminus 
 \overline{S_0^\e(\omega)}$ and the corresponding set indicatior functions by $\chi_0^\e(\omega)$ and $\chi_1^\e(\omega)$. 
 
For each $\omega\in\Lambda,$ we consider the following Dirichlet problem in $S:$ for $\lambda<0$ and $f^{\e} \in L^2(S),$ find $u^\e \in W^{1,2}_0(S)$ such that 
 \begin{equation}\label{misha3}
 \int_{S} A^\e(\cdot, \omega)\nabla u^{\varepsilon}\cdot \nabla v -\lambda\int_S u^{\varepsilon} \cdot v=\int_{S} f^{\varepsilon}  v\qquad\qquad \forall v\in W^{1,2}_0(S),
 \end{equation}
where
$$
A^\e(\cdot,\omega)= \chi_1^\e(\omega)A_1
+\varepsilon^2\chi_0^\e(\omega)I,
\qquad\qquad 
\omega\in\Lambda, 
$$
with a symmetric and positive-definite matrix $A_1.$
 
 For all $\omega \in \Lambda$ we also define the Dirichlet operator 
 $\mathcal{A}^{\e}(\omega)$  in $L^2(S)$ corresponding to the differential expression 
 $-\div A^{\e} (\cdot,\omega) \nabla u,$ {\it e.g.} 
 by considering the 
 bilinear form 
 \[
 \int_{S}A^{\e} (\cdot,\omega)\nabla u\cdot\nabla v,\qquad u, v\in W^{1,2}_0(S).
 \]
It is well known that the spectrum of $\mathcal{A}^{\e}(\omega)$ is discrete. The following subspace of $W^{1,2}(\Omega)$ will play an essential role in our analysis:
$$  W_{0}^{1,2}(\mathcal{O})=\bigl\{v \in W^{1,2}(\Omega): 
v(T_x \omega)=0 \textrm{ on }  \mathbf{R}^n \backslash \mathcal{O}_{\omega}\ \ \forall\omega\in\Lambda\bigr\}. $$
Notice that as a consequence of Ergodic Theorem (Theorem \ref{thmergodic}) one has
$$  W_0^{1,2}(\mathcal{O})=\bigl\{ v \in W^{1,2}(\Omega): \chi_{\mathcal{O}} v=v\bigr\},$$
{\it i.e.} $W_0^{1,2}(\mathcal{O})$ consists of $W^{1,2}$-functions that vanish on $\Omega \backslash \mathcal{O}.$ Henceforth we assume that $\omega\in\Lambda$ without mentioning it explicitly.
 
The next two lemmas use a standard smoothening (or ``mollification") procedure, which we now describe. We take $g\in L^2(\Omega)$ and ({\it cf.} \cite[p.\,232]{zhikov2}) choose a nonnegative even function $\rho \in C_0^{\infty} (\mathbf{R}^n)$ with $\int_{\mathbf{R}^n} \rho=1,$ $\supp\rho \subset B_1(0)$ and write $\rho_{\delta}(x)=\delta^{-n} \rho(x/\delta)$ for
all $\delta>0$. For each $\delta>0,$ we define the regularisation ${\mathcal R}_\delta[g]$ of $g$ by  
$$ 
{\mathcal R}_{\delta}[g] (\omega)=\int_{\mathbf{R}^n} \rho_{\delta} (y) g(T_y \omega)dy,
=\int_{\mathbf{R}^n} \rho_{\delta} (y) g(T_{-y} \omega)dy,\qquad\omega\in\Omega.
$$
Notice that 
\begin{equation} \label{nada10} 
{\mathcal R}_{\delta}\bigl[g\bigr](T_x\omega)=\int_{\mathbf{R}^n} \rho_{\delta} (y) g(T_{x-y} \omega)dy=\int_{\mathbf{R}^n} \rho_{\delta} (x-y) g(T_y \omega)dy=\int_{\mathbf{R}^n} \rho_{\delta} (y-x) g(T_y\omega)dy,
\end{equation}
from which we infer that
$$ 
D_i {\mathcal R}_{\delta}[g](\omega)=-\int_{\mathbf{R}^n} \partial_i \rho_{\delta} (y)g(T_y \omega)dy,\quad\omega\in\Omega,\qquad i=1,2,\dots,n. 
$$
Arguing by induction, we show that ${\mathcal R}_{\delta}[g] \in W^{\infty,2}(\Omega),$ and if $g \in L^{\infty} (\Omega)$ then ${\mathcal R}_{\delta}[g] \in \mathcal{C}^{\infty} (\Omega).$ Before we state and prove the lemmas, we introduce additional notation.  We define the space 
$$ \mathcal{C}_0^{\infty} (\mathcal{O}):=\bigl\{v \in \mathcal{C}^{\infty} (\Omega): v= 0 \textrm{ on } \Omega \backslash \mathcal{O}\bigr\},$$
as well as the sets 
$$D^{k,m}_{\omega}:=\biggl\{x \in  \mathcal{O}^k_{\omega}: \dist\bigl(x, \partial \mathcal{O}^k_{\omega}\bigr)>\frac{1}{m}\biggr\},\quad\quad k,m \in \mathbf{N}.$$
Also, for all $m \in \mathbf{N}$ we define the set 
$$ B^m:=\{\omega \in \Omega: 0 \in D^{k,m}_{\omega} \textrm{ for some } k \in \mathbf{N}   \}\subset\Omega. $$ 	
By using the density of $\mathbf{Q}^n$ in $\mathbf{R}^n$ it can be seen that for all $m \in \mathbf{N}$ the set $B^m$ is measurable.
Notice that for each fixed $\omega,k,m,$ where $m$ is large enough, there exist constants $C_1,C_2>0$ such that 
\begin{equation}
\frac{C_1}{m}\bigl |\mathcal{O}^k_{\omega}\bigr| \leq\bigl|\Omega^k_{\omega} \backslash D^{k,m}_{\omega}\bigr| \leq \frac{C_2}{m}\bigl|\mathcal{O}^k_{\omega}\bigr|.
\label{set_bounds}
\end{equation}
In the next lemma we assume that a relaxed version of the right inequality in (\ref{set_bounds}) holds uniformly in $\omega.$ 
\begin{lemma} \label{lemmagloria}
Suppose that for a.e. $\omega \in \Omega$ there exists a sequence of positive values $C_m$ converging to zero, such that 
$$  
\quad\bigl|\mathcal{O}^k_{\omega} \backslash D^{k,m}_{\omega}\bigr| \leq C_m\bigl|\mathcal{O}^k_{\omega}\bigr|\quad\quad \forall k \in \mathbf{N}.    
$$
Then the set
	$\mathcal{C}_0^{\infty} (\mathcal{O})$
	is dense in  $L^2(\mathcal{O})$. 
\end{lemma}
\begin{proof}
 By using Ergodic Theorem and the assumption of the lemma, it can be shown that $P(\mathcal{O} \backslash B^m) \to 0$ as $m \to \infty$. To prove  the density, it suffices to approximate $g:=\chi_{B^m} f$, where $f \in L^{\infty}(\mathcal{O})$ by a function from $\mathcal{C}_0^{\infty} (\mathcal{O}),$ for which we 
 use the above mollification procedure. Notice that for $\delta>0$ small enough, one has ${\mathcal R}_{\delta}[g] \in \mathcal{C}_0^{\infty}(\mathcal{O})$. 
It remains to check ${\mathcal R}_{\delta}[g] \to g$ as $\delta \to 0$, but this follows from the strong continuity of the group $U(x),$ see \eqref{matthaus1000}: 
\begin{eqnarray*}
\bigl\|g-{\mathcal R}_{\delta}[g]\bigr\|_{L^2(\Omega)} \leq \int_{\mathbf{R}^n} \rho_{\delta}(y)\bigl\|U(y) g-g\bigr\|_{L^2(\Omega)}dy \to 0, 
\end{eqnarray*}
as required.
\end{proof} 	
Notice that, by the standard Poincar\'{e} inequality, for each $D^{k,m}_{\omega}$ there exists $C>0$ such that
\begin{equation}\label{nada1} 
 \int_{\mathcal{O}^k_{\omega} \backslash D^{k,m}_{\omega}}u^2 dx \leq C\bigl|\mathcal{O}^k_{\omega} \backslash D^{k,m}_{\omega}\bigr|^2 \int_{\mathcal{O}^k_{\omega} \backslash D^{k,m}_{\omega}} |\nabla u|^2 dx\qquad \forall u \in W_0^{1,2} (\mathcal{O}^k_{\omega}). 
 \end{equation} 
In the following lemma we impose this condition uniformly. 
\begin{lemma}
 \label{marko110} 
	Assume that for a.e. $\omega \in \Omega$ there exists a constant $C>0$ such that 
	\begin{equation} \label{nada2} 
		\bigl|\mathcal{O}^k_{\omega} \backslash D^{k,m}_{\omega}\bigr| \leq  \frac{C}{m}\bigl|\mathcal{O}^k_{\omega}\bigr|\qquad\forall k, m,
		\end{equation} 
and that \eqref{nada1} is satisfied for all $k$ and large enough $m.$ Then the set $\mathcal{C}_0^{\infty} (\mathcal{O})$ is dense in $W_0^{1,2} (\mathcal{O})$.
\end{lemma}
\begin{proof}
	We take $f \in W_0^{1,2} (\mathcal{O})$ and define
	$f_M:=\chi_{|f|\leq M} f+\chi_{|f|\geq M}M$. Notice that as a consequence of \eqref{matthaus0}, $f_M \in W_0^{1,2} (\mathcal{O})$ and $D_i f_M=\chi_{|f|< M} D_i f $.
It is easily seen that $f_M \to f$ in $W^{1,2}(\Omega)$ as $M \to \infty$. 
	 Thus we can assume, without loss of generality,  that $f \in L^{\infty}(\Omega) \cap W_0^{1,2}(\mathcal{O})$. We define 
	 $h_m=K_{1/2m}[\chi_{B^m}f]$. 
	 It can be seen from the proof of Lemma \ref{lemmagloria} that $h_m \to f$ in $L^2(\mathcal{O})$ as $m \to \infty$. 
	 Notice that for a.e. $\omega \in B^{\lceil 2m/3 \rceil}$ we have for $i=1,\dots,n$
	$$  D_i h_m (\omega)= K_{1/2m}[D_if] (\omega)=\int_{\mathbf{R}^n} \rho_{1/2m} (y) D_i f(T_y \omega)dy,$$
	 and therefore
	  $$ \| D_i h_m-D_i f\|_{L^2(B^{\lceil 2m/3 \rceil})}\leq \bigl\| K_{1/2m}[D_if]-D_i f\bigr\|_{L^2(\Omega)} \to 0. $$
	  
	  Notice also that  for a.e. $\omega \in \Omega$ there exist $C_1,C_2>0$ such that for all $k,m \in \mathbf{N},$ where $m$ is sufficiently large, we have
	 $$
	 \bigl\|\partial_i h_m(\cdot,\omega)\bigr\|^2_{L^2(\mathcal{O}^k_{\omega} \backslash D^{k,\lceil 2m/3 \rceil}_{\omega})} \leq C_1m^2 \|f\|^2_{L^2(\mathcal{O}^k_{\omega} \backslash D^{k,\lceil m/2 \rceil}_{\omega})} \leq C_2\|f\|^2_{W^{1,2}(\mathcal{O}^k_{\omega} \backslash D^{k,\lceil m/2 \rceil}_{\omega})} ,  $$  
	 where we have used \eqref{nada10}, \eqref{nada1}, \eqref{nada2} and  Young's inequality. By using the Ergodic Theorem we conclude that there exists $C>0$ such that
	 $$ \|D_i h_m\|^2_{L^2(\mathcal{O} \backslash B^{\lceil 2m/3 \rceil})} \leq C \|f\|^2_{W^{1,2}(\mathcal{O} \backslash B^{\lceil m/2 \rceil})}, $$
from which the claim follows. 
\end{proof}

\section{Limit equations and two-scale resolvent convergence}
\label{lim_eq}

We define the quadratic form 
\begin{equation*}
	 A_1^{\rm hom}\xi \cdot \xi:= \inf_{\varphi \in W^{1,2}(\Omega)}\int_{\Omega \backslash \mathcal{O}} A_1(\xi+\nabla_\omega \varphi) \cdot (\xi+\nabla_\omega \varphi),\qquad\xi\in\mathbf{R}^n,	
\end{equation*}	 
and denote by $\mathcal{W}^{1,2}(\Omega \backslash \mathcal{O})$  the completion of $W^{1,2}(\Omega)$ with respect to the seminorm 
$\| \nabla_{\omega} \varphi\|_{L^2(\Omega \backslash \mathcal{O})},$\linebreak
$\varphi\in W^{1,2}(\Omega).$ 
The proof of the following lemma is straightforward. 
\begin{lemma} \label{characteristic} 	For each $\xi \in \mathbf{R}^n$ there exists  $p_{\xi} \in \mathcal{W}^{1,2} (\Omega \backslash \mathcal{O})$ such that 
\begin{equation*}
 A_1^{\rm hom} \xi \cdot \xi=\int_{\Omega \backslash \mathcal{O}} A_1(\xi+p_{\xi}) \cdot (\xi+p_{\xi}),
\end{equation*}
or, equivalently,
\begin{equation} \label{tambaca3} 
\int_{\Omega \backslash \mathcal{O}} A_1(\xi+p_{\xi})\cdot \nabla_{\omega} \varphi=0 \qquad \forall \varphi  \in \mathcal{C}^{\infty} (\Omega).
\end{equation}
In particular, one has 
$ 
A_1^{\rm hom}\le A_1.
$
\end{lemma}
\begin{remark} \label{tambaca2}
It follows from the observations in  \cite[p.\,265--266]{zhikov1} that if the following extension property is satisfied for a.e. $\omega \in \Omega$: for all $u\in C_0^{\infty} (B_1(0))$ there exists $\varepsilon^0_{\omega}>0$ and  a sequence of functions $(\widetilde u_{\varepsilon})$ such that
\begin{equation*}
\widetilde u_{\varepsilon}=u \textrm{ in } B_1(0) \backslash\bigcup_{k \in \mathbf{N}} \varepsilon \mathcal{O}^k_{\omega}\textrm{ for } \varepsilon \leq \varepsilon^0_{\omega}, 
\quad\quad\qquad \int_{B_2(0)} |\nabla \widetilde u_{\varepsilon}| \leq C_{\omega} \int_{B_1(0)
\backslash\bigcup_{k\in\mathbf{N}}\varepsilon\mathcal{O}^k_{\omega}}  |\nabla u|^2, 
\end{equation*}
where $C_{\omega}$ is a constant independent of $u$ and $\varepsilon,$ then the matrix $A_1^{\rm hom}$ is positive definite. 
\end{remark} 
Notice that under Assumption \ref{kirill100}, the extension property in Remark \ref{tambaca2} is satisfied. 
We define the space 
\begin{equation*}
H:=L^2(S)+\bigl\{u\in L^2(S\times\Omega): u\vert_{S\times(\Omega\setminus\mathcal{O})}=0\bigr\},
\label{direct_sum}
\end{equation*}
which is clearly a direct sum, naturally embedded in $L^2(S\times \Omega)$. Before stating the next theorem we prove a simple lemma that implies that 
gives norm bounds for each component of $H$
by the norm in $L^2(S\times \Omega).$
\begin{lemma} \label{kirill101} 
Let $f_0 \in L^2(S)$ and $f_1 \in L^2(S  \times \Omega)$ such that $f_1 \equiv 0$ on $\Omega \backslash \mathcal{O}$. Then there exists a constant $C>0$ such that 
$$ \|f_0\|_{L^2(S)}+\|f_1\|_{L^2(S \times \Omega)} \leq C\|f_0+f_1\|_{L^2(S \times \Omega)}, 
$$
where we use the natural embedding $L^2(S) \hookrightarrow L^2(S \times \Omega)$. 
\end{lemma} 
\begin{proof}
Notice that by Cauchy-Schwartz inequality we have 
\begin{eqnarray*}
2\int_{S \times \Omega}\bigl|f_0(x)\bigr|\bigl|f_1(x,\omega)\bigr| dx dP(\omega) &=&2\int_{S}\bigl|f_0(x)\bigr| \int_\mathcal{O}\bigl|f_1(x,\cdot)\bigr| dP dx \\ &\leq&\| 
f_0\|^2_{L^2(S)}+\left\|\int_\mathcal{O}f_1\right\|^2_{L^2(S)}
\leq \| f_0\|^2_{L^2(S)}+P(\mathcal{O}) \|f_1\|^2_{L^2(S \times \Omega)},
\end{eqnarray*}
and hence
\begin{equation}
\|f_0+f_1\|^2_{L^2(S \times \Omega)} \geq \|f_0\|^2_{L^2(S \times \Omega)}-2\int_{S \times \Omega} |f_0| |f_1| dx dP +\|f_1\|^2_{L^2(S \times \Omega)} \\
\geq \bigl(1-P(\mathcal{O})\bigr) \|f_1\|^2_{L^2(S \times \Omega)}.
\end{equation}
It remains to bound $\|f_0\|_{L^2(S \times \Omega)}$ by $\|f_0+f_1\|_{L^2(S \times \Omega)}$, which is done by the triangle inequality:  
\begin{equation*}
\|f_0\|^2_{L^2(S)} \leq 2\bigl(\|f_0+f_1\|^2_{L^2(S)}+\|f_1\|^2_{L^2(S \times \Omega)}\bigr) 
\leq\dfrac{2\bigl(2-P(\mathcal{O})\bigr)}{1-P(\mathcal{O})} \|f_0+f_1\|^2_{L^2(S \times \Omega)}.
\end{equation*}
\end{proof}
By $\mathcal{P}: L^2(S \times \Omega) \to H$ we denote the orthogonal projection.
 For $f \in L^2(S \times \Omega)$ we have
$$ \mathcal{P}f(x,\omega) =\int_{\Omega \backslash \mathcal{O}} f(x,\cdot)dP + \chi_\mathcal{O}(\omega) \left(f(x,\omega)-\int_{\Omega \backslash \mathcal{O}} f(x,\cdot)dP\right). $$
\begin{theorem} \label{misha10}
Under Assumption \ref{kirill100}, let $\lambda<0$ and suppose that $(f^{\eps})$ be a bounded sequence in $L^2(S)$ such that 
$f^{\eps} \wtwoscale f\in L^2(S \times \Omega).$ For each $\varepsilon>0,$ consider the solution $u^\e$ to \eqref{misha3}.  Then  for a.e. $\omega \in \Omega$ one has $u^\e \wtwoscale u_0+u_1(\cdot,\omega)$, where $u_0 \in W_0^{1,2}(S),$ $u_1 \in L^2(S, W_0^{1,2} (\mathcal{O}))$ satisfy 
\begin{align}
&\int_S A_1^{\rm hom} \nabla u_0\cdot\nabla \varphi_0 -\lambda \int_S\bigl(u_0+\langle u_1 \rangle_\Omega\bigr) \varphi_0 = \int_S \langle f \rangle_\Omega\,\varphi_0
\quad\quad \forall \varphi_0 \in W_0^{1,2}(S), \label{mik1} \\[0.6em]
 &\int_\mathcal{O} \nabla_{\omega} u_1 (x, \cdot)\cdot\nabla_{\omega} \varphi_1 -\lambda  \int_\mathcal{O}\bigl(u_0(x)+u_1(x,\cdot)\bigr)  \varphi_1 = \int_\mathcal{O} f(x,\cdot)  \varphi_1
\quad\quad\forall  \varphi_1 \in W_0^{1,2}(\mathcal{O}). \label{mik2}
\end{align}

\end{theorem} 
\begin{remark}
\label{marko1002} 
The system \eqref{mik1}--\eqref{mik2} is understood in the weak sense:
\begin{align} 
\label{zadnje} 
	\int_S A_1^{\rm hom} \nabla u_0\cdot\nabla \varphi_0&+\int_{S \times \Omega} \nabla_{\omega} u_1 \cdot\nabla_{\omega} \varphi_1  -\lambda \int_{S \times \Omega} (u_0+ u_1) (\varphi_0+\varphi_1)\\[0.3em]  \nonumber &  \hspace{+10ex}=\int_{S \times \Omega}   f  (\varphi_0+\varphi_1) \quad\quad \forall \varphi_0 \in W_0^{1,2}(S),\  \varphi_1 \in L^2\bigl(S,W_0^{1,2}(\mathcal{O})\bigr).
\end{align}
Noting that $W_0^{1,2}(\mathcal{O})$ is a closed subspace of $W^{1,2}(\Omega)$ and bearing in mind Lemma \ref{kirill101}, it follows by the Lax-Milgram lemma that for all $f\in L^2(S \times \Omega),$ $x \in S$ the problem (\ref{zadnje})
has a unique solution in $W_0^{1,2}(\mathcal{O}).$  
Its solutions for the right-hand sides $f \in L^2(S \times \Omega)$ and $\mathcal{P}f$ coincide. 
The solution of the equation \eqref{mik2} has the form 
\begin{equation} 
u_1(x, \omega)=u(x,\omega)+u_0(x)w( \omega), 
\label{u1_rep}
\end{equation}
where $u \in L^2(S, W_0^{1,2} (\mathcal{O}))$ is the solution of the equation \eqref{mik2} obtained by setting $u_0=0$ and $w \in W_0^{1,2}(\mathcal{O})$ is the solution of the equation \eqref{mik2} obtained by setting $u_0=1$ and $f=0.$ Substituting \eqref{u1_rep} into \eqref{mik1}, we obtain an equation on $u_0$. 
\end{remark} 
\begin{proof} 
The proof follows a standard argument. First, by \eqref{misha3}, there exists a constant $C>0$ such that
\begin{equation}
\|\nabla u^{\e}\|_{L^2(S_1^{\e})}+\e \|\nabla u^{\e}\|_{L^2 (S_0^{\e})}+\|u^{\e}\|_{L^2(S)}  \leq C. 
\end{equation}
For each $\e>0$ we extend $u^{\e}|_{S_1^{\e}},$ 
using Assumption \ref{kirill100}, to a sequence $\widetilde u^{\e},$ which is bounded in $W^{1,2}(S).$ From Proposition \ref{properties} we infer that there exist $u_0 \in W^{1,2} (S)$, $u_1 \in L^2(S,  W^{1,2}(\Omega))$, $u_2 \in L^2( S, \mathcal{W}^{1,2}(\Omega))$ such that on a subsequence we have 
\begin{equation}
\widetilde u^{\e} \to  u_0 \textrm{ strongly in } L^2(S),
\quad\quad
\nabla \widetilde u^{\e} \wtwoscale  \nabla u_0+ \nabla_{\omega} u_2,\qquad
u^{\e}-\widetilde u^{\e}  \wtwoscale u_1, 
\quad\quad
\e \nabla (u^{\e}-\widetilde u^{\e})  \wtwoscale  \nabla_{\omega} u_1. 
\end{equation}
To obtain the equation \eqref{mik1}, we take test functions of the form $\varphi_0(x)+\e a(x) \varphi(T_{\varepsilon^{-1}x}\omega)$ in \eqref{misha3},  where $\varphi_0 \in W_0^{1,2} (S)$,  $\varphi \in  W^{1,2}(\Omega)$ and $a \in C_0^1(S)$. In the limit as $\varepsilon\to0$ we obtain
\begin{equation}
\int_S  \int_{\Omega \backslash \mathcal{O}} A_1 (\nabla u_0+\nabla_{\omega} u_2) \cdot\bigl(\nabla \varphi_0+ a\nabla_{\omega} \varphi_1\bigr) dP dx -\lambda  \int_S\bigl(u_0+\langle u_1 \rangle_\Omega\bigr) \varphi_0 
 =\int_S  \langle f \rangle_\Omega\,\varphi_0.\label{tambaca10} 
\end{equation}
Setting $\varphi_0=0,$ it follows that 
$$
 \int_{\Omega \backslash \mathcal{O}} A_1 (\nabla u_0(x)+\nabla_{\omega} u_2) \cdot  \nabla_{\omega} \varphi_1 dP=0 \qquad \textrm{a.e.}\  x \in S,
$$
and the characterisation \eqref{tambaca3} yields $\nabla_{\omega} u_2(x,\omega)=p_{\nabla u_0(x)}(\omega)$ a.e. $x \in S,$ $\omega\in\Omega.$ Taking arbitrary $\varphi_0 \in W_0^{1,2}(S)$ in \eqref{tambaca10}, we obtain the ``macroscopic'' part \eqref{mik1} of the limit ptoblem. The ``microscopic'' part \eqref{mik2} is obtained by taking test functions of the form $a(x) \varphi(T_{\e^{-1}x} \omega )$ in \eqref{misha3}, where $a \in C_0^1(S)$, $\varphi \in W_0^{1,2} (\mathcal{O})$. The convergence of the whole sequence can be deduced by uniqueness of the solution of the system \eqref{mik1}--\eqref{mik2}. 
\end{proof} 
\begin{remark}
	The following observation was made in \cite{Zhikov2000} in the periodic setting. The formulation \eqref{zadnje} can be interpreted from the operator-theoretic point of view. Namely, we define a positive definite operator $\mathcal{A}$ on a dense linear subset of $V=W_0^{1,2}(S)+L^2(S,W_0^{1,2}(\mathcal{O}))$  (which is a dense subset of $H$ under the condition of Lemma \ref{lemmagloria}), as follows. 
	One takes $\lambda<0$ and defines the domain ${\rm dom}(\mathcal{A})$ as the set of solutions of \eqref{zadnje} obtained for varying $f \in H$. 
	To see that  ${\rm dom}(\mathcal{A})$ is dense in $H,$ take the solutions $u_0+u_1, w_0+w_1\in V$ for $f, g \in H,$ respectively.  
	Setting $\varphi_0=u_0$, $\varphi_1=u_1$ as the test function in the equation for $(w_0, w_1)$ and $\varphi_0=w_0$, $\varphi_1=w_1$ as the test function in the equation for $(u_0,u_1)$ yields 
	$$ \int_{S \times \Omega} f(w_0+w_1)=\int_{S \times \Omega} g(u_0+u_1). $$
	Thus, if $g \perp u_0+u_1$ then necessarily $w_0+w_1=0,$ which implies $g=0$. The operator $\mathcal{A}: {\rm dom}(\mathcal{A})\to H$ defined by 
	$\mathcal{A}(u_0+u_1)=f+\lambda(u_0+u_1)$
	is a bounded linear mapping between Hilbert spaces, where the norm on ${\rm dom}(\mathcal{A})$ is given by
	$$ \|u_0 + u_1\|^2_{{\rm dom}(\mathcal{A})} =\bigl\|\mathcal{A}(u_0+u_1)\bigr\|^2_H+\|u_0+u_1\|^2_V. $$

\end{remark}
We shall need the following statement for the convergence of spectra of the operators associated with (\ref{misha3}). 
It is proved in the same way as the previous theorem, and we omit the proof.
\begin{proposition} \label{kirill200} Under Assumption \ref{kirill100}, let $\lambda<0$ and suppose that $(f^{\eps})$ is a bounded sequence in $L^2(S_0^{\e})$ such that $\chi_0^{\e} f^{\eps} \wtwoscale f\in L^2(S \times \Omega).$ For each $\varepsilon>0,$ let $z^\e \in W_0^{1,2} (S_0^{\e})$    be the solution of 
\begin{equation} \label{igor1000}
\varepsilon^2 \int_{S_0^{\e}} \nabla z^{\e} \cdot \nabla v - \lambda \int_{S_0^{\e}} z^{\e}v=\int_{S_0^{\e}} f^{\e} v \quad\quad \forall v \in W_0^{1,2} (S_0^{\e}).
\end{equation} 
Then  for a.e. $\omega \in \Omega$ one has $z^\e \wtwoscale z(\cdot,\omega)$, where   $z \in L^2(S, W_0^{1,2}(\mathcal{O}))$ is the solution of the problem
 \begin{equation} \label{mira100} 
  \int_\mathcal{O} \nabla_{\omega} z(x,\cdot) \cdot \nabla_{\omega} v-\lambda \int_\mathcal{O} z(x,\cdot) v= \int_\mathcal{O} f(x,\cdot)  v \quad\quad \forall v \in W_0^{1,2}(\mathcal{O}).  
  \end{equation}
\end{proposition}  
\begin{remark} \label{mira1000}
Theorem \ref{misha10} and Proposition \ref{kirill200}	are still valid if, instead of a fixed $\lambda<0,$ we take a sequence $(\lambda^{\e}) \subseteq \mathbf{R}$ such that $\lambda_\e \to \lambda \in \mathbf{R}$ and $\liminf_{\e \to 0} \dist(\lambda^{\e}, \Sp \mathcal{A}^{\e})>0$, for Theorem  \ref{misha10}, i.e. $\liminf_{\e \to 0} \dist(\lambda^{\e}, \Sp \mathcal{T}^{\e})>0$ for Proposition \ref{kirill200}, where 
$\mathcal{T}^{\e}:=-\e^2 \Delta_{S_0^{\e}}.$
Notice that $\Sp \mathcal{T}^{\e}$ splits into the spectra of scaled Laplace operators on each hole contained in $S_0^{\e}:$ 
\begin{equation} \label{kirill1000} 
\Sp  \mathcal{T}^{\e}= \bigcup_{n \in K^\varepsilon_{\omega} } \Sp\bigl(-\Delta_{\mathcal{O}^n_{\omega}}\bigr), 
\end{equation} 
where $K^\varepsilon_{\omega}$ is defined in \eqref{K_def}.
Notice that there exists $C>0$  
such that for all $\lambda \in \mathbf{R}$ the solution $u^{\e}$ of \eqref{misha3} satisfies
\begin{equation} \label{igor100}
\|\nabla u^{\e}\|_{L^2(S_1^{\e})}+\e \|\nabla u^{\e}\|_{L^2 (S_0^{\e})}+\|u^{\e}\|_{L^2(S)}  \leq C\left(
\dist(\lambda, \Sp \mathcal{A}^{\e})^{-1}+\lambda+1\right)\|f^{\e}\|_{L^2(S)},
\end{equation}
and similarly the solution of \eqref{igor1000} satisfies 
\begin{equation} \label{igor10000}
\e \|\nabla z^{\e}\|_{L^2 (S_0^{\e})}+\|z^{\e}\|_{L^2(S_0^{\e})}  \leq C\left(
\dist(\lambda, \Sp \mathcal{T}^{\e})^{-1}+\lambda+1\right)\|f^{\e}\|_{L^2(S)}.
\end{equation}
\end{remark}

In what follows we denote by $-\Delta_\omega$ the operator generated by the bilinear form 
\begin{equation}
\int_\mathcal{O} \nabla_{\omega} u\cdot\nabla_{\omega}v,
\quad\quad u, v\in W_0^{1,2}(\mathcal{O}). 
\label{delta_omega_form}
\end{equation}
As a consequence of Proposition \ref{kirill200} and Remark \ref{mira1000}, we have the following statement. 
\begin{corollary}  
 \label{marko101} 
Assume that Assumption \ref{kirill100} holds. Then 
$$ 
\Sp(-\Delta_{\omega}) \subseteq \overline{\bigcup_{n \in \mathbf{N} } \Sp\bigl(-\Delta_{\mathcal{O}^n_{\omega}}\bigr)}\qquad {\rm a.e.}\ \omega \in \Omega.
$$
\end{corollary} 
\begin{proof}
	Take $ \lambda \notin \overline{\bigcup_{n \in \mathbf{N} } \Sp (-\Delta_{\mathcal{O}^n_{\omega}})}$ and $f \in L^2(\mathcal{O}),$ and define $f^{\e}(x,\omega):=\chi_0^{\e} f(T_{\e^{-1}x}\omega)\stwoscale f$. As a consequence of Remark \ref{mira1000}, the sequence of solutions of \eqref{igor1000} converges weakly two-scale to the solution of \eqref{mira100}, which is a resolvent equation. Moreover, \eqref{kirill1000} and \eqref{igor10000}  imply the existence of $C>0$ such that
	$$\|u_1\|_{L^2(S \times \Omega)} \leq \liminf_{\e \to 0} \|z^{\e}\|_{L^2(S_0^{\e})}\leq C\|f\|_{L^2(\mathcal{O})}, $$
and therefore $\lambda \notin  \Sp(-\Delta_{\omega})$. 
\end{proof} 

\section{Spectral completeness for inclusions}
\label{inclusions_completeness}

In this part we prove that 
$$
\overline{\bigcup_{n \in \mathbf{N} } \Sp\bigl(-\Delta_{\mathcal{O}^n_{\omega}}\bigr)} \subseteq   \Sp(-\Delta_{\omega})\qquad {\rm a.e.}\ \omega \in \Omega. 
$$
We shall use the assumptions of Lemma \ref{marko110}
as well as assume that for each $\lambda_0>0$ there exists $M_{\lambda_0}>0$ such that for a.e. $\omega \in \Omega$ the following implication holds:
\begin{equation} 
-\Delta u=\lambda u, \quad u\in W_0^{1,2} (\mathcal{O}^k_{\omega}), \textrm{ for some } k \in \mathbf{N}, \lambda \leq \lambda_0 \implies \|u\|_{L^{\infty}(\mathcal{O}^k_{\omega})} \leq M_{\lambda_0}\|u\|_{L^2(\mathcal{O}^k_{\omega})}.
\label{implication}
\end{equation} 
Notice that, 
by regularity theory, the above condition is satisfied for a fixed $\omega \in \Omega$ and $k \in \mathbf{N},$  
whenever the boundary $\partial \mathcal{O}^k_{\omega}$ is sufficiently regular. 
 In what follows we use a sequence $\{\widetilde \varphi^k\}_{k \in \mathbf{N}} \subset C_0^{\infty} ([0,c_2+1]^n) $ that is dense in $W_0^{1,2}  ([0,c_2+1]^n),$ where the constant $c_2$ is defined in Assumption \ref{kirill100}. 

We will now define a sequence of random variables that is invariant for all $\omega \in \mathcal{O}$ whose realisation  is such that the shape that contains the origin is the same.  For $q=(q_1,\dots,q_n) \in \mathbf{Q}^n$ define the set 
$$ \mathcal{O}_q:=\bigl\{\omega \in \mathcal{O}:\textrm{there exists } k_0 \in \mathbf{N} \textrm{ such that } q\in\mathcal{O}^{k_0}_{\omega}\bigr\}. $$
\begin{lemma} \label{lucia1} 
For every $q \in \mathbf{Q}^n$, $\mathcal{O}_{q} \subset \Omega$ is measurable. 
\end{lemma} 
\begin{proof}  
Notice that 
\begin{eqnarray}\label{connection1}
\omega \in  \mathcal{O}_q &\iff&\textrm{There exists a polygonal line that connects $0$ and $q$ and } \\ \nonumber & & \textrm{consists of a finite set of straight segments with rational endpoints } \\  \nonumber & & \textrm{such that for all 
	$l \in \mathbf{Q}^n$ on this line one has }
\textrm{$T_l \omega \in \mathcal{O}$.} 
\end{eqnarray} 
Since for each fixed $q \in \mathbf{Q}^n$ there is a countable set of lines satisfying the property \eqref{connection1}, the set $\mathcal{O}_q$ is measurable. 	
\end{proof} 	
We define the random variables
$$ 
D_i(\omega):=\inf\{q_i: \ \omega \in \mathcal{O}_q\},\qquad \omega\in\Omega,\qquad\qquad i=1, \dots,n. 
$$
Notice that  $D_i=+\infty$ whenever $\omega \notin \mathcal{O},$ and also, due to the assumption, $-c_2 \leq D_i\leq 0$ for a.e. $\omega \in \mathcal{O}.$ We denote by $D$ the random vector 
\begin{equation}
D:=-(D_1, \dots, D_n)+\biggl(\frac{1}{2},\dots,\frac{1}{2}\biggr).
\label{D_def}
\end{equation}

For a.e. $\omega \in \mathcal{O}$, $m \in \mathbf{N}$ we define the set 
\begin{equation*}
	D^m_{\omega}:=
	\biggl\{ x \in \mathbf{R}^n: \textrm{ there exists } k_0 \in \mathbf{N} \textrm{ such that }  x \in  \mathcal{O}^{k_0}_{\omega} \textrm{ and }
	\dist\bigl(x, \partial \mathcal{O}^{k_0}_{\omega}\bigr)> \frac{1}{m}\biggr\}. 
\end{equation*}
Furthermore, we introduce the set  $U_{\omega}\subset [0,c_2+1]^n$, which is a translation of the set  $\mathcal{O}^{k_0}_{\omega}$ containing the origin:
\begin{equation*}
U_{\omega}:=
\big\{x \in [0,c_2+1]^n:
x -D \in \mathcal{O}^{k_0}_{\omega}  \textrm{ for } k_0 \in \mathbf{N} \textrm{ such that }  0\in  \mathcal{O}^{k_0}_{\omega}\big\}. 
\end{equation*}
Finally, we define a characteristic function of the translation of the set $D^m_{\omega}$ and  a measurable function of $\o$ taking values in $W_0^{1,2}([0,c_2+1]^n)$

	\begin{equation}
	\chi^{m}(x,\omega ):=\chi_{D^m_{\omega}}(x -D),\qquad
	\varphi^{k,m}(x,\omega):=\rho_{1/2m} *\left(\chi^m(x,\omega )\widetilde \varphi^k(x)\right).
	\end{equation}

Notice that for a.e. $\omega \in \mathcal{O}$ one has $\supp \varphi^{k,m}(\cdot,\omega) \subset U_{\omega}$.

\begin{lemma} \label{ante10}
	For every $k,m \in \mathbf{N}$, the function $\omega \mapsto \varphi^{k,m}(\cdot,\omega)$ taking values in $W_0^{1,2}([0,c_2+1]^n)$ is measurable with respect to the Borel $\sigma$-algebra on $W_0^{1,2}([0,c_2+1]^n)$. 
\end{lemma}

\begin{proof}
Firstly notice that 
\begin{equation} \label{ante2} 
\omega \mapsto \chi^{m}(\cdot, \omega)\widetilde \varphi^k(\cdot), 
\end{equation} 
is a measurable mapping taking values in the set $L^2([0,c_2+1]^n)$, with Borel $\sigma$-algebra. To check this notice that for each $q \in \mathbf{Q}^n$ the set 
$$ B_q:=\bigl\{\omega \in \Omega: q \in D^m_{\omega}\bigr\},$$ 	
is measurable: the related proof is similar to that of Lemma \ref{lucia1}. 
Further, for $\psi \in C_0^{\infty} (\mathbf{R}^n)$ the norm $\| \psi- \chi^{m} \widetilde \varphi^k\|_{L ^2({\mathbf R}^n)}$ is written as a limit of Riemann sums, and each Riemann sum can be written in terms of a finite number of $\chi_{B_q}$ and values of  function $\widetilde \varphi^k(\cdot)$. Thus $\omega \mapsto\| \psi- \chi^{m} \widetilde \varphi^k\|_{L ^2({\mathbf R}^n)}$ is measurable.  
Since the topology in $L^2(\mathbf{R}^n)$ is generated by the balls of the form $B(\psi,r)$, where $\psi \in C_0^{\infty} (\mathbf{R}^n)$ and $r \in \mathbf{Q}$ we have that the mapping given by \eqref{ante2} is measurable. The final claim follows by using the fact that the convolution is a continuous (and thus measurable) operator from $L^2$ to $W^{1,2}$. 
\end{proof}
Notice that by construction  $\{\varphi^{k,m}(\cdot,\omega)\}_{k,m \in\mathbf{N}}$ $\subset  C_0^{\infty} (U_{\omega})$ is a dense subset of $W_0^{1,2} (U_{\omega})$ for a.e. $\omega \in \Omega$ (see also the proof of Lemma \ref{marko110}). For $0\le a\le b$ we introduce the following subset of $\mathcal{O}:$
\begin{equation}
 E_{a,b}:=\bigl\{\omega \in\mathcal{O}:  
 -\Delta_{\mathcal{O}^{k_0}_{\omega}}\ \textrm{has an eigenvalue in\ }  [ a, b ] \textrm{ for } k_0 \in \mathbf{N} \textrm{ such that } 
 0 \in \mathcal{O}^{k_0}_{\omega}\bigr\}.
 \label{E_def}
 \end{equation}
 For $0\leq a\leq b$ and a.e. $\omega \in E_{a,b}$ we also define $S_{a,b,\omega} \subset W_0^{1,2} (U_{\omega})$ as follows: 
  \begin{eqnarray*}
  S_{a,b,\omega}&:=&\bigl\{\psi \in W_0^{1,2}(U_{\omega}): 
   \psi\ \textrm{is an eigenfunction of } -\Delta_{U_{\omega}} \textrm{ whose eigenvalue is in } 
  [ a, b ]  \bigr\}. 
\end{eqnarray*} 	
Finally, for every $r \in \mathbf{R}$ and $k,m \in \mathbf{N}$ we define the random variable 
\begin{equation}
X_r^{k,m}(\omega):= \left\{\begin{array}{lr}
 \dfrac{\bigl\|-\Delta \varphi^{k,m}(\cdot,\omega)-r\varphi^{k,m}(\cdot,\omega)\bigr\|_{W^{-1,2 }(U_{\omega})}}{\bigl\|\varphi^{k,m}(\cdot,\omega)\bigr\|_{L^2(U_{\omega})}} & \textrm{if } \varphi^{k,m}(\cdot,\omega)\neq 0, \\[0.9em] +\infty & \textrm{otherwise.} \end{array} \right. 
 \label{X_indices_def}
 \end{equation}
\begin{lemma} \label{ante1} 
	For every $r \in \mathbf{R}$ and $k,m \in \mathbf{N}$, $X_r^{k,m}$ is a measurable function.  
\end{lemma}
\begin{proof}
We use Lemma \ref{ante10} and the fact that
  $-\Delta$ is a continuous map from $W^{1,2}$ to $W^{-1,2}$ and $\|\cdot\|_{W^{-1,2}(U_{\omega})}$ is a measurable function, since 
 $$\|\psi(\cdot,\omega)\|_{W^{-1,2}(U_{\omega})}:=\sup_{k,m\in \mathbf{N}}\left\{\frac{_{W^{-1,2}(U_{\omega})}\big\langle \psi(\cdot,\omega), \varphi^{k,m}(\cdot,\omega) \big\rangle_{W^{1,2}_0(U_{\omega})} }{\|\varphi^{k,m}(\cdot,\omega)\|_{W^{1,2}(U_\omega)}}: \varphi^{k,m}(\cdot,\omega) \neq 0 \right\}. $$	
\end{proof}
\begin{lemma}
	For $0\leq a\leq b$, the set $E_{a,b}$ is measurable. 
\end{lemma}
\begin{proof}
	The claim follows by observing that
	$$ E_{a,b}=\Bigl\{\omega\in\mathcal{O}: \inf_{k,m \in \mathbf{N}, r \in \mathbf{Q} \cap [ a,b]}X^{k,m}_r(\omega)=0\Bigr\}.$$
\end{proof}

Now we are going to define a measurable mapping from $\mathcal{O}$ to the subspace $S_{a,b,\omega}$. We set it to be an $L^2$-projection onto $S_{a,b,\omega}$ of a specially chosen function of $x$ and $\o$. We need the following measurability lemma.
\begin{lemma}
	Assume that $\omega \mapsto \varphi(\cdot, \omega)$ is a measurable function taking values in $L^2(U_{\omega})$ for a.e. $\omega \in E_{a, b}$. Then the $L^2$-distance 
	$ 
	\omega \mapsto \dist_{L^2(U_\omega)}\bigl(\varphi(\cdot,\omega), S_{a,b,\omega}\bigr),$ 
	$\omega\in E_{a,b},$
	is a measurable map. 
\end{lemma}
\begin{proof}
	The claim follows from the formula
\begin{equation*}
	\dist_{L^2(U_\omega)}\bigl(\varphi(\cdot,\omega), S_{a,b,\omega}\bigr)
	=\limsup_{n\to \infty} \inf_{k,m \in \mathbf{N}}\biggl\{\bigl\|\varphi^{k,m}(\cdot,\omega)-\varphi(\cdot,\omega)\bigr\|_{L^2(U_\omega)}:X_r^{k,m}(\omega) < \frac{1}{n}
	\textrm{ for some } r \in \mathbf{Q} \cap [a,b]\biggr\}. 
\end{equation*}
\end{proof}
For $0\leq a\leq b$ and $\omega\in E_{a, b}$ we define  the measurable map $\omega \mapsto \varphi_{a,b}(\cdot,\omega)$ as follows: 
$$ 
\varphi_{a,b}(\cdot,\omega)=\varphi^{k_0(\omega), m_0( \omega)} (\cdot,\omega),          
$$
where
\begin{align}
k_0(\omega):= \min_{k \in \mathbf{N}}\biggl\{k: \dist_{L^2(U_\omega)}\bigl(\varphi^{k,m}(\cdot,\omega), S_{a,b,\omega}\bigr)&\neq \bigl\|\varphi^{k,m}(\cdot,\omega)\bigr\|_{L^2(U_\omega)},\nonumber\\
&\frac{1}{2} \leq\bigl\|\varphi^{k,m}(\cdot,\omega)\bigr\|_{L^2(U_\omega)}\leq 1 \textrm{ for some } m \in \mathbf{N}\biggr\}, 
\label{k0_def}
\end{align}
and $m_0(\omega)$ is the minimal value of $m$ in (\ref{k0_def})
Notice that in this way for a.e. $\omega \in E_{a, b}$ the $L^2$-projection of $\varphi_{a,b}(\cdot, \omega)$ on $S_{a,b,\omega}$ is not zero. We also define the random variable $R: \Omega \to[0,+ \infty)$ in the following way:
$$ 
R(\omega):= \left\{\begin{array}{ll} \dist\bigl(0, \partial \mathcal{O}^{k_0}_{\omega}\bigr) & \textrm{ if } 0 \in \mathcal{O}^{k_0}_{\omega} \textrm{ for some } k_0 \in \mathbf{N},\\[0.3em] 0 & \textrm{ otherwise}.  \end{array} \right.
$$
By invoking the measurability of $\mathcal{O}_q,$ $q\in{\mathbf Q}^n,$ see  Lemma \ref{lucia1}, 
it is easily seen that $R$ is indeed measurable. Next, for $0\leq a\leq b$, $l>0$ we define the random variable $\psi_{a,b,l}: \Omega \to \mathbf{R}$ by 
$$ 
\psi_{a,b,l}(\omega):= \left\{\begin{array}{ll} \limsup_{n \to \infty}\fint_{B(D,\min\{l,R(\omega)\})}\varphi^{k_1(\omega,n),m_1(\omega,n)}(\cdot, \omega) & \textrm{ if } R(\omega)>0,\ \omega\in E_{a, b},  \\[0.4em] 0 & \textrm{ otherwise},  \end{array} \right.
$$
where, for all $n\in{\mathbf N},$ 
\begin{align*}
	k_1(\omega,n):=& \min_{k \in \mathbf{N}} \biggl\{k: X^{k,m}_r < \frac{1}{n}  \textrm{ for some }  r \in \mathbf{Q} \cap [a,b],\\  
	&\bigl\|\varphi^{k,m}(\cdot, \omega)
	-\varphi_{a,b}(\cdot, \omega)\bigr\|_{L^2(U_\omega)} < \dist_{L^2(U_\omega)}\bigl(\varphi_{a,b}(\cdot,\omega), S_{\a,b,\omega}\bigr)+\frac{1}{n} 
	\textrm{ for some } m \in \mathbf{N}\biggr\},
\end{align*}
$m_1(\omega)$ is the corresponding minimal value\footnote{The function $\varphi^{k,m}(\cdot, \omega)$ is an ``approximate eigenfunction'' for $-\Delta_{U_\omega},$ see (\ref{X_indices_def}).} of $m,$
and $B(D,\min\{l,R(\omega)\})$ is the ball with the centre at $D$ and radius $\min\{l,R(\omega)\},$ see (\ref{lucia1}).
We also define 
$$\psi_{a,b}:=\limsup_{l^{-1}\in\mathbf{N},\ l\to 0} \psi_{a,b,l}.  $$
Notice that in this way $\psi_{a,b}$ is the value at the origin (taking into account for $\omega \in\mathcal{O}$ the relative position of the origin with respect to the shape) of the (unique) $L^2$-projection of $\varphi_{a,b}$ onto $S_{a,b,\omega}.$  As a consequence of \eqref{implication}, we have $|\psi_{a,b}| \leq M_b. $ Notice that by construction $\psi_{a,b}\neq 0$ if $P(E_{a,b})>0$.

We are ready for the proof of main statement.
\begin{theorem}
\label{marko1000} 
Under Assumption \ref{kirill100}, the assumption of Lemma \ref{marko110} and \eqref{implication},  one has 
	$$
	\overline{\bigcup_{n \in \mathbf{N} } \Sp\bigl(-\Delta_{\mathcal{O}^n_{\omega}}\bigr)} \subseteq   \Sp(-\Delta_{\omega})\qquad {\rm a.e.}\ \omega \in \Omega.
	$$
\end{theorem}
\begin{proof}
We take $l \geq 0$. There are two possibilities:
\begin{enumerate}
		\item There exists $\varepsilon>0$ such that $E_{l-\varepsilon, l+\varepsilon}$ has zero probability. In this case we denote 
		$$ 
		\varepsilon_0(l):= \sup_{\varepsilon>0}\bigl\{\varepsilon: P(E_{l-\varepsilon, l+\varepsilon})=0\bigr\}.
		$$
	\item For all $\varepsilon>0$ the set $E_{l-\varepsilon, l+\varepsilon}$ has positive probability. 
\end{enumerate} 	
In the case (a), by the continuity of probability, we conclude that 
$P\bigl(E^0_{ l-\varepsilon_0(l), l+\varepsilon_0(l)}\bigr)=0,$ where ({\it cf.} (\ref{E_def}))
\begin{equation*}
  E^0_{a,b}:=\bigl\{\omega \in \mathcal{O}:  
  -\Delta_{\mathcal{O}^{k_0}_{\omega}}\ \textrm{has an eigenvalue in\ } (a, b) \textrm{ for } k_0 \in \mathbf{N} \textrm{ such that }  0 \in \mathcal{O}^{k_0}_{\omega}\bigr\}.  
\end{equation*}
 
By Lemma \ref{marko100} and Corollary \ref{marko101} we infer that 
$$
\bigl(l-\varepsilon_0(l), l+\varepsilon_0(l)\bigr)\subseteq {\mathbf C}\setminus\overline{\bigcup_{n \in \mathbf{N} } \Sp\bigl(-\Delta_{\mathcal{O}^n_{\omega}}\bigr)}\subseteq{\mathbf C}\setminus\Sp(-\Delta_{\omega})\qquad{\rm a.e.\ } \omega \in\Omega.
$$	
In particular, we conclude that $l \notin \Sp(-\Delta_{\omega})$. 

In the case (b) we construct a Weyl sequence showing that $l \in  \Sp(-\Delta_{\omega})$. To this end, we define 
$$
\psi_n:=
\bigl\|\psi_{l-1/n,l+1/n}\bigr\|_{L^2(\mathcal{O})}^{-1}\psi_{l-1/n,l+1/n},\qquad n\in{\mathbf N}.
$$ 
Then, by the above construction and using Ergodic Theorem,
one has 
$$
\bigl\|-\Delta_{\omega}\psi_n-l\psi_n\bigr\|_{L^2(\mathcal{O})} \leq \frac{1}{n},\qquad n\in{\mathbf N}.
$$

It follows from the above that  $\Sp(-\Delta_{\omega})$ consists of exactly those $l \in \mathbf{R}$ that satisfy the property (b). The set $\Sp(-\Delta_{\omega})$ is closed, hence its complement is a countable union of open disjoint intervals. 
Every element of such an interval $(d_1, d_2)$ satisfies the property (a) with $l=(d_1+d_2)/2,$
$\varepsilon_0(l)=(d_2-d_1)/2,$ and therefore $P(E^0_{d_1, d_2})=0.$  
Using Lemma \ref{marko100}, we obtain
$$
(d_1, d_2)\subseteq{\mathbf C}\setminus\overline{\bigcup_{n \in \mathbf{N} } \Sp\bigl(-\Delta_{\mathcal{O}^n_{\omega}}\bigr)}\qquad {\rm a.e.}\ \omega \in \Omega. 
$$
The claim follows since there is only countable number of such intervals. 
\end{proof}

\section{Convergence of spectrum} 
\label{spectral_convergence}

In our analysis we keep in mind the examples set in Section \ref{examples_sec}, for which it is shown that  
$ \Sp (-\Delta_{\omega}) \subseteq \Sp(\mathcal{A})$.
In the present section we assume that this holds, as well as the conclusion of Theorem \ref{marko1000}, {\it i.e.} 
\begin{equation} \label{igor100000} 
\Sp(-\Delta_{\omega})=  \overline{\bigcup_{n \in \mathbf{N} } \Sp\bigl(-\Delta_{\mathcal{O}^n_{\omega}}\bigr)}. 
\end{equation}

We are interested in approximating the spectra $\Sp \mathcal{A}^{\e}(\omega)$ of the operators $\mathcal{A}^{\varepsilon}(\omega)$ (see Section \ref{problem_form}) by the spectrum $\Sp\mathcal{A}$ of the limit operator. 
We claim that $\Sp \mathcal{A}^{\e}(\omega)\to\Sp\mathcal{A}$  for a.e. $\omega \in \Omega,$ 
where the convergence is understood in the Hausdorff sense:
\begin{enumerate}
	\item For all $\lambda \in 	\Sp \mathcal{A}$ there are $\lambda^{\e} \in \Sp  \mathcal{A}^{\e}(\omega)$ such that $\lambda^{\e} \to \lambda$. 
	\item If $\lambda_{e} \in \Sp \mathcal{A}^{\e}(\omega)$ and $\lambda^{\e} \to \lambda$, then $\lambda \in \Sp \mathcal{A}$. 
\end{enumerate}	
We prove this claim by adapting the argument of \cite{zhikov2005}. First, we introduce the notion of strong resolvent convergence. 

\begin{definition}
	Let $\mathcal{A}^{\e}(\omega)$ and $\mathcal{A}$ be the operators acting on $L^2(S)$ and on $H\subset L^2(S \times \Omega),$ respectively. We say that $\mathcal{A}^{\e}(\omega)$ strongly two-scale resolvent converge to $\mathcal{A}$ and write $\mathcal{A}^{\e} \stwoscale \mathcal{A}$ if
	$$ 
	f^{\e} \stwoscale f, \quad f \in L^2(S\times \Omega) \implies \bigl(\mathcal{A}^{\e}(\omega)+I\bigr)^{-1} f^{\e} \stwoscale (\mathcal{A}+I)^{-1}f\qquad{\rm for\ a.e.}\ \omega\in\Omega.                   
	 $$
\end{definition} 
It can be shown that the property (a) is satisfied if we have strong two-scale resolvent convergence (see the proof of \cite[Proposition 2.2]{zhikov2005}). Theorem \ref{misha10} shows that the following implication holds:
$$ 
f^{\e} \wtwoscale f, \quad f \in L^2(S\times \Omega) \implies\bigl(\mathcal{A}^{\e}(\omega)+I\bigr)^{-1} f^{\e} \wtwoscale (\mathcal{A}+I)^{-1}f.                    
$$
It can be shown that this is equivalent to strong two-scale resolvent convergence (see \cite[Proposition 2.8]{zhikov2005}) and thus the property (a) is satisfied. 

In order to prove (b), we start from the eigenvalue problem of the operator $\mathcal{A}^{\e}(\omega)$ (it has a compact resolvent and its spectrum is discrete), {\it i.e.}, of the solutions to
\begin{equation} \label{kirill10000} 
 \mathcal{A}^{\e}(\omega)u^{\e} =s^{\e} u^{\e}, \quad \int_S (u^{\e})^2=1. 
 \end{equation}
If we have that $s_\e \to s$
and $u^{\e} \wtwoscale u$, then we would also have $\mathcal{A} u=s u$. However, the problem would be if $u=0$, because then $s \notin \Sp \mathcal{A}$. The next lemma tells us if $s \notin \Sp (-\Delta_{\omega})$
then necessarily the sequence of eigenvalues are compact with respect to strong two-scale converegence and thus $s$ belongs to the point spectrum of the operator $\mathcal{A}$, since then necessarily $u\neq 0$. 

 


\begin{theorem}
Suppose that \eqref{igor100000} holds and that for each $\e>0$, $(s^{\e},u^{\e})$ satisfy \eqref{kirill10000}. If $s^{\e} \to s \notin \Sp (-\Delta_{\omega})$, then for a.e. $\omega \in \Omega$ the sequence $(u^{\e})$ is compact in the sense of strong two-scale convergence.
\end{theorem} 	
\begin{proof}
$u^{\e} \in W_0^{1,2}(S_1^{\e})$ satisfies
$$ \int_{S_1^{\e}} A_1 \nabla u^{\e} \cdot \nabla v+\e^2 \int_{S_0^{\e}} \nabla u^{\e}\cdot\nabla v= s^{\e} \int_{S} u^{\e} v \quad \forall v \in  W_0^{1,2}(S). $$
We use Assumption \ref{kirill100} and for each $\varepsilon$ extend $u^{\e}|_{S_0^{\e}}$, denoting the extensions by $\widetilde u^{\e}$. Notice that there exists $C>0$ such that
\begin{equation} 
\label{bukal1} 
 \widetilde u^{\e} \in W_0^{1,2} (S), \quad \Delta \widetilde u^{\e}=0 \textrm{ on } S_0^{\e}  \quad   \|\widetilde u^{\e}\|_{W^{1,2}(S)} \leq C.
 \end{equation} 
The difference $z^{\e}:=u^{\e}-\widetilde u^{\e}$ satisfies: 
\begin{equation}
z^{\e} \in W_0^{1,2} (S_0^{\e}),\qquad
 \e^2\int_{S_0^{\e}} \nabla z^{\e} \cdot \nabla v-s^{\e}\int_{S_0^{\e}}z^{\e} v = s^{\e} \int_{S_0^{\e}} \widetilde u^{\e}v\quad\quad
\forall v \in W_0^{1,2}(S_0^{\e}). 
\label{nada100000}
\end{equation}
From the estimate \eqref{bukal1} we see that $(\widetilde u^{\e})$ is weakly compact in $W_0^{1,2} (S)$ and thus there exists $\widetilde u \in W_0^{1,2} (S)$ such that $\widetilde u^{\e} \rightharpoonup \widetilde u,$ which immediately implies
$s_\e \chi_0^{\e} \widetilde u^{\e} \stwoscale s  \widetilde u \chi_{\mathcal{O}} (\omega).$ 
Furthermore, as a consequence of \eqref{kirill1000}, \eqref{igor10000} and \eqref{igor100000}, the following estimate holds for some $C>0:$
\begin{equation*}
\e \|\nabla z^{\e} \|_{L^2(S_0^{\e})}+\|z^{\e}\|_{L^2(S_0^{\e})} \leq C. 
\end{equation*}
Therefore, from Proposition \ref{kirill200} and Remark \ref{mira1000} we conclude that $z^{\e} \wtwoscale z \in L^2(S,W_0^{1,2}(\mathcal{O}))$ where the limit $z$ satisfies 
\begin{equation} 
\int_\mathcal{O} \nabla_{\omega} z(x,\cdot) \cdot \nabla_{\omega} v-s \int_\mathcal{O} z(x,\cdot) v=s \int_\mathcal{O} \widetilde u (x)  v \quad \forall v \in W_0^{1,2}(\mathcal{O}).  
\label{nada1000}
\end{equation}
 We also consider the problem

\begin{equation}
 m^{\e} \in W_0^{1,2} (S_0^{\e}),\qquad
\e^2\int_{S_0^{\e}} \nabla m^{\e} \cdot \nabla v-s^{\e}\int_{S_0^{\e}}m^{\e} v = s^{\e} \int_{S_0^{\e}} z^{\e}v \label{bukal10000}\quad\quad
  \forall v \in W_0^{1,2}(S_0^{\e}). 
  \end{equation}
In the same way as before we conclude that for some $C>0$: 
\begin{equation*}
\e \|\nabla m^{\e} \|_{L^2(S_0^{\e})}+\|m^{\e}\|_{L^2(S_0^{\e})} \leq C. 
\end{equation*}
Analogously, we conclude that $m^{\e} \wtwoscale m \in L^2(S,W_0^{1,2}(\mathcal{O}))$ which satisfies 
\begin{equation} \label{limitm1}
\int_\mathcal{O} \nabla_{\omega} m(x,\cdot) \cdot \nabla_{\omega} v-s \int_\mathcal{O} m(x,\cdot) v=s \int_\mathcal{O} z (x,\cdot) v \quad \forall v \in W_0^{1,2}(\mathcal{O}).  
\end{equation} 
By testing \eqref{nada100000} with $m^{\e}$ and \eqref{bukal10000} with $z^{\e}$ we conclude
$$\lim_{\e \to 0} \int_{S_0^{\e}} (z^{\e})^2= \lim_{\e \to 0} \int_{S_0^{\e}}\widetilde u^{\e} m^{\e} =\int_{S \times \Omega} \widetilde u m. $$ 
Finally, by testing \eqref{nada1000} with $m(x,\cdot)$ and \eqref{limitm1} with $z(x,\cdot)$ and integrating over $S$ we conclude
$$ \int_{S \times \Omega} \widetilde u m=\int_{S \times \Omega} z^2, $$ 
which completes the proof.
\end{proof}

\section{Spectrum of the limit operator: examples} 
\label{examples_sec}

This section is devoted to the description of the spectrum of the limit operator. Since it crucially depends on the intrinsic properties of the microscopic part of the operator and the properties of the probability space, it does not seem feasible (at least at the current stage of research in this area) to provide a characterisation of the spectrum in a general setting. We shall consider several interesting, from the point of view of applications, examples of probability spaces and configurations of soft inclusions. The general example of a finite number of shapes of randomly varying size is described in Section \ref{misha20}. Then we consider the case of a single shape of fixed size in Section \ref{simple_example}, and the case of a single shape of randomly varying size in Section \ref{sec:7.3}, for which we provide the full description of the spectrum of the limit operator with the proofs. The characterisation of the spectrum in the general case of Section \ref{misha20} is analogous to the case of a single shape considered in Section \ref{sec:7.3}.

\subsection{The setting of finite number of shapes of varying size} \label{misha20}
Let $(\widetilde{\omega}_j)_{j\in \mathbf{Z}^n}$ be a sequence of  independent and identically distributed random vectors taking values in $\mathbf{N}_0^l \times [r_1, r_2],$ where $0<r_1\le r_2\le1$ and $(\widetilde{\Omega},\widetilde{\mathcal{F}},\widetilde{P} )$ is an appropriate probability space. We also assume that we have a finite number of shapes $Y_k\subset Y,$ $k \in \mathbf{N}_0^l,$ that represent the inclusions, 
where the first and the second components of $\widetilde \o_j = (k_j, r_j)$ model the shape and the size,
 respectively. We also set $Y_0= \emptyset$. 
On $\widetilde{\Omega}$ there is a natural shift $\widetilde{T}_z (\widetilde{\omega}_j)=(\widetilde{\omega}_{j-z})$, which is ergodic. We next state the discrete analogue of Lemma \ref{marko100}. 
\begin{lemma}
	Assume that $\widetilde{\Omega}_0 \subseteq \widetilde{\Omega}$ is a set of full measure. Then there exists a subset $\widetilde{\Omega}_1 \subseteq \widetilde{\Omega}_0$ of full measure such that for each $\widetilde{\omega} \in \widetilde{\Omega}_1$, $z  \in \mathbf{Z}^n$ we have $\widetilde{T}_z \widetilde{\omega} \in \widetilde{\Omega}_0$. 	 
\end{lemma}

We treat $Y:=[0,1)^n$ as a probability space with Lebesgue measure $dy$ and the standard algebra $\mathcal{L}$ of Lebesgue measurable sets, and define 
$$
\Omega=\widetilde{\Omega} \times Y,\quad \mathcal{F}= \widetilde{\mathcal{F}} \times \mathcal{L} ,\quad P= \widetilde{P} \times dy. 
$$
On $\Omega$ we define a dynamical system $T_x (\widetilde \omega, y)=(\widetilde T_{[x+y]} \widetilde \omega, x+y-[x+y])$. By $\mathcal{O} \subseteq \Omega$ we define the set $\mathcal{O}=\{(\widetilde \omega, y):\, \widetilde \o_0\in \mathbf{N}_0^l \times [r_1, r_2],\, y\in r_0 Y_{k_0}\}.$
It is easily seen that $\mathcal{O}$ is measurable. For a fixed $\o = (\widetilde{\o},y),$ the realisation $\mathcal O_\o$ consists of the inclusions $r_j Y_{k_j} + j -y,\,j\in\Z^n$. Next, we describe the generators $D_i,$ $i=1,2,\dots, n,$ in the present example. Taking $f \in W^{1,2}(\Omega)$
and using the above lemma, note that there exists a subset of full measure $\widetilde{\Omega}_1 \subseteq \widetilde{\Omega}$ such that for all
$\widetilde \omega \in \widetilde \Omega_1$ and $z \in \mathbf{Z}^n$ we have  $f(\widetilde T_z \widetilde \omega,\cdot) \in W^{1,2}(Y)$. It is clear that for $x\in Y+z-y$ one has $f(x,\o):= f(T_x\o) = f(\widetilde T_z \widetilde \omega,x-(z-y))$. Using this fact and the statement following (\ref{matthaus0}), we infer that
\begin{align*}
	W^{1,2}(\Omega)=\Bigl\{f\in L^2(\widetilde \Omega \times Y): \textrm{ for a.e. } \widetilde\omega \in \widetilde \Omega, \  f(\widetilde \omega, \cdot)&\in W^{1,2}(Y),\\[0.1em] f(\widetilde T_{z+e_k}\widetilde \omega,\cdot)\big|_{\{y_k=0\}}&=f(\widetilde T_{z}\widetilde \omega,\cdot)\big|_{\{y_k=1\}}\ \forall z \in \mathbf{Z}^n, k \in \{1,\dots,n\}\Bigr\}.
\end{align*} 
and
\begin{equation} \label{marko1001} 
	(D_i f)(\widetilde \omega,y)=\partial_{y_i} f(\widetilde \omega,y),\qquad i=1,2,\dots, n.
	\end{equation}

\subsection{Simple example} 
\label{simple_example}

In this section we set $l=0,$ $r_1=r_2=1,$ so that $\mathbf{N}_0^l \times [r_1, r_2]=\{0,1\},$ and, by a standard procedure, see {\it e.g.} \cite{Shiryaev1996}, identify the elements of the probability space $\widetilde{\Omega}$ with
sequences $\widetilde \omega = (\widetilde\o_z)_{z\in \mathbf{Z}^n}$ whose components $\widetilde\o_z$
take values in the two-element set $\{0,1\}$.  
Let $Y_1$ be an open subset of $Y$ whose closure is contained in $Y$ (``soft inclusion") . The value $0$ or $1$ of $\widetilde\o_z,$ $z\in\mathbf{Z}^n,$ corresponds to the absence or the presence of the inclusion in the ``shifted cell" $Y+z,$ respectively. We also set
$$
\mathcal{O} =\bigl\{\o=(\widetilde{\o},y)\,:\, \widetilde \o_0 = 1,\, y\in Y_1\bigr\}\subseteq \O.
$$ 
Then, for a given  $\omega=(\widetilde{\omega}, y)\in\Omega,$ the realisation $\mathcal{O}_\o = \{ x\,:\, T_x(\widetilde \o, y) \in \mathcal{O}\}$ is the union of the sets (``inclusions") $Y_1+z-y$ over all $z\in\mathbf{Z}^n$ such that $\widetilde \o_z = 1$. For this example the space $W^{1,2}_0(\mathcal{O})$ consists of all functions of the form
\begin{equation}\label{10a}
v(\o) = v(\widetilde \o,y) = \left\{ \begin{array}{ll}
v_{\widetilde \o}(y), & (\widetilde \o,y)\in \mathcal{O},\mbox{ where } v_{\widetilde \o}\in W^{1,2}_0(Y_1),
\\[0.4em]
0 & \mbox{otherwise}.
\end{array}\right.
\end{equation}
It is also important to understand how one applies the stochastic gradient. For a function $v(\o) \in W^{1,2}_0(\mathcal{O})\subseteq L^2(\O)$ we have (see \eqref{marko1001})
\begin{equation}
\nabla_{\omega} v = \left\{ \begin{array}{ll}
\nabla_y v_{\widetilde \o}(y), & (\widetilde \o,y)\in \mathcal{O},\mbox{ where } v_{\widetilde \o}\in W^{1,2}_0(Y_1),
\\[0.4em]
0 & \mbox{otherwise}.
\end{array}\right.
\end{equation}

Consider formally the spectral problem for the limit operator:
\begin{align}
\label{16} 
&\int_S A_1^{\textrm{\rm hom}} \nabla u_0\cdot\nabla \varphi_0 = \l \int_S\bigl(u_0+\langle u_1 \rangle\bigr) \varphi_0  \qquad\qquad \forall \varphi_0 \in W^{1,2}_0(S), \\[0.7em]
&\int_\mathcal{O} \nabla_{\omega} u_1 (x, \cdot) \nabla_{\omega} \varphi_1 = \l \int_\mathcal{O}\bigl(u_0(x)+u_1(x,\cdot)\bigr)  \varphi_1 \qquad\qquad  \forall  \varphi_1 \in W^{1,2}_0(\mathcal{O}). 
\label{16_bis}
\end{align}
We write the solution to the ``microscopic" equation (\ref{16_bis}) in the form $u_1(x,\o) = \lambda u_0(x) v(\o)$, where (recall Remark \ref{marko1002}) 
\begin{equation}
v\in  W^{1,2}_0(\mathcal{O}),\quad\qquad-\Delta_\o v  =  \lambda v + 1.
\end{equation}
In other words, $v$ is given by (\ref{10a}) with $v_{\widetilde \o}(y) $ satisfying 
\begin{equation}\label{19}
-\Delta_y v_{\widetilde \o}(y)  =  \lambda v_{\widetilde \o}(y) + 1,\qquad y\in Y_1,
\end{equation}
whenever $\widetilde \o$ such that $\widetilde\o_0 = 1$ and $v_{\widetilde \o} = 0$ otherwise. 

We label the eigenvalues of the operator in (\ref{19}) in the increasing order, where we repeat multiple eigenvalues, so that $\nu_j,$ $j\in{\mathbf N},$ and $\nu'_j,$ $j\in{\mathbf N},$ are, respectively, the eigenvalues whose eigenfunctions $\varphi_j$ have non-zero integral over $Y_1$ and the eigenvalues whose eigenfunctions $\varphi'_j$ have zero integral over $Y_1.$
Following \cite{Zhikov2000}, we write the solution to (\ref{19}) via the spectral decomposition
\begin{equation}\label{15}
v_{\widetilde \o}=  \sum_{j=1}^\infty
(\nu_j-\lambda)^{-1}\biggl(\int_{Y_1}\varphi_j\biggr)\varphi_j,
\end{equation}
and thereby 
\begin{equation}
\langle v \rangle_\O =  P\bigl(\{\widetilde\o:\,\widetilde\o_0=1\}\bigr)\int_{Y_1} v_{\widetilde \o} dy = P\bigl(\{\widetilde\o:\,\widetilde\o_0=1\}\bigr) \sum_{j=1}^\infty
(\nu_j-\lambda)^{-1}\biggl(\int_{Y_1}\varphi_j\biggr)^2.
\end{equation}

Substituting the obtained representation for $u_1$ into the ``macroscopic" equation (\ref{16}) yields
\begin{equation}
-\div A^{\rm hom}_1 \nabla u_0 = \beta(\lambda) u_0,\quad u_0 \in W^{1,2}_0(S),
\label{mac_eq}
\end{equation}
where
\begin{equation}
\beta(\lambda):= \lambda \big(1 + \lambda \langle v \rangle_\O
\big) = \lambda+\lambda^2P\bigl(\{\widetilde\o:\,\widetilde\o_0=1\}\bigr) \sum_{j=1}^\infty
(\nu_j-\lambda)^{-1}\biggl(\int_{Y_1}\varphi_j\biggr)^2
\label{beta}
\end{equation}
is a stochastic version of the ``Zhikov function'' $\beta$ in \cite{Zhikov2000}.


Assume for the moment that $S = \R^n$. Then the intervals where $\beta(\l)\geq 0$ are the ``spectral bands" of $\mathcal A,$ and additionally a Bloch-type spectrum is given by 
$\{\nu_j': j\in {\mathbf N}\}$. The set $\{\l:\, \beta(\l) <0\} \setminus
\{\nu_j': j\in {\mathbf N}\}$ corresponds to the gaps in the spectrum of $\mathcal A$. 

In the setting of this paper, namely, for a bounded domain $S\subset{\mathbf R}^n,$ instead of each spectral band $\beta(\l)\geq 0$ lying to the left of $\nu_j$ we have a ``band'' of discrete spectrum: a countable set of eigenvalues 
\begin{equation}
\bigl\{\l_{j,k}:\, \nu_{j-1} < \l_{j,k} < \nu_j,\,\beta(\l_{j,k}) = \mu_k\bigr\},
\label{lambdas}
\end{equation}
with the accumulation point at the right end  $\nu_j$ of each band, where $\mu_k$ are the eigenvalues of the operator  $-\div A^{\rm hom}_1 \nabla$ defined by the form
\begin{equation}
 \int_{S}A^{\rm hom}_1\nabla u\cdot\nabla v,\qquad u, v\in W^{1,2}_0(S).
 \label{A1_form}
 \end{equation}
 The Bloch-type spectrum of $\mathcal A$ consists of eigenvalues $\nu_j'$ of infinite multiplicity with eigenfunctions of the form $f(x)v_j'(\o)$ with $f\in L^2(S)$ and 
\begin{equation}
v_j'(\o) = v_j'(\widetilde\o,y)= \left\{ \begin{array}{ll}
\varphi_j'(y), & (\widetilde \o,y)\in \mathcal{O},
\\[0.4em]
0 & \mbox{otherwise}.
\end{array}\right.
\end{equation}
Summarising, the spectrum of $\mathcal A$ is given by
\begin{equation}
\sigma(\mathcal A) = \Big(\bigcup_j\bigl\{\nu_j', 
\nu_j\bigr\}\Big) \cup
\bigl\{\l_{j,k}: j,k\in{\mathbf N}\bigr\}.
\end{equation}

\subsection{More advanced example} 
\label{sec:7.3}

Here 
we allow the inclusions to randomly change size, so that $l=0,$ $0<r_1<r_2<1.$ 
By analogy with the previous section, we assume that $\widetilde \O$ consists of sequences 
$\widetilde \omega=(\widetilde\o_z)_{z\in \mathbf{Z}^n}$ such that $\widetilde\o_z \in \{0\}\cup [r_1,r_2],$ $z\in{\mathbf Z}^n.$ We also assume that the restriction to to $[r_1,r_2]$ of the probability measure  on $\widetilde{\Omega}$ 
is absolutely continuous with respect to Lebesgue measure. As before, consider $Y_1\subset Y,$ and denote by $Y_{1,r} : = r (Y_1 - y^{\rm c}) + y^{\rm c}$, where $y^{\rm c}$ is the centre of $Y$, the ``scaled inclusion", requiring that $\overline{Y_{1,r_2}} \subset Y$, in order for the extension property in Assumption \ref{kirill100} to hold. The values $0$ or $r\in [r_1,r_2]$ of $\widetilde\o_z$ correspond to the absence of an inclusion or the presence of the inclusion $Y_{1,r}$ in the cell $Y+z,$ respectively. Furthermore, define $\mathcal{O}: = \{\o=(\widetilde{\o},y)\,:\, y\in Y_{1,\widetilde{\o}_0}\}\subseteq \O$. Then a realisation $\mathcal{O}_\o = \{ x\,:\, T_x(\widetilde \o, y) \in \mathcal{O}\}$ is the union of the inclusions $Y_{1,\widetilde{\o}_z}+z-y$ for all 
$z\in \Z^n,$ where in the case $\widetilde \o_z = 0$ we set $Y_{1,\widetilde{\o}_z}=\emptyset.$ The space $W^{1,2}_0(\mathcal{O})$ consists of functions of the form
\begin{equation}
\label{10}
v(\o) = v(\widetilde \o,y) = \left\{ \begin{array}{ll}
v_{\widetilde \o}(y), & (\widetilde \o,y)\in \mathcal{O},\mbox{ where } v_{\widetilde \o}\in W^{1,2}_0(Y_{1,\widetilde{\o}_0}),
\\[0.4em]
0, & \mbox{otherwise}.
\end{array}\right.
\end{equation} 

Consider the spectral problem for $u_1,$ namely 
\begin{equation*}
\int_\mathcal{O} \nabla_{\omega} u_1 (x, \cdot) \nabla_{\omega} \varphi = \lambda \int_\mathcal{O}\bigl(u_0(x)+u_1(x,\cdot)\bigr)  \varphi\qquad \forall\varphi\in W^{1,2}_0(\O),
\end{equation*}
and separate the variables, as in Section \ref{simple_example}: $u_1(x,\o) = \lambda u_0(x) v(\o),$ where the function $v$ satisfies
\begin{equation}
\label{omega_grad}
\int_\mathcal{O} \nabla_{\omega} v\cdot\nabla_{\omega} \varphi =  \int_\mathcal{O} (1+\lambda v) \varphi\qquad \forall\varphi\in W^{1,2}_0(\O).
\end{equation}
The stochastic gradient is given by
\begin{equation*}
\nabla_{\omega} v = \left\{ \begin{array}{ll}
\nabla_y v_{\widetilde \o}(y), & (\widetilde \o,y)\in \mathcal{O},\mbox{ where } v_{\widetilde \o}\in W^{1,2}_0(Y_{1,\widetilde{\o}_0}),
\\[0.4em]
0, & \mbox{otherwise},
\end{array}\right.
\end{equation*}
and therefore the problem (\ref{omega_grad})  is equivalent to 
\begin{equation}\label{23}
\int\limits_{\{{\widetilde{\o}_0\in[r_1,r_2]\}}} \int\limits_{Y_{1,\widetilde{\o}_0}} \nabla_y v \nabla_y \varphi \,dy \,dP(\widetilde{\o}) =  \int\limits_{\{{\widetilde{\o}_0\in[r_1,r_2]\}}} \int\limits_{Y_{1,\widetilde{\o}_0}}  (1+\lambda v)  \varphi \,dy \,dP(\widetilde{\o}).
\end{equation}

For each $r\in[r_1, r_2],$ the eigenvalues $\nu_{j,r}, \nu_{j,r}'$  and (orthonormal) eigenfunctions $\varphi_{j,r}, \varphi_{j,r}'$ of the operator $-\Delta_y$ acting in $W^{1,2}_0(Y_{1,r})$ are obtained by scaling the   eigenvalues and eigenfunctions of  $-\Delta_y$ acting in $W^{1,2}_0(Y_{1})$, in particular, $\nu_{j,r} = r^{-2} \nu_j$, $\nu_{j,r}' = r^{-2} \nu_j'$. Therefore, the formula (\ref{15}) with  $\nu_j,$ $\varphi_j$  replaced by $\nu_{j,r},$ $\varphi_{j,r}$ gives the solution to 
\begin{equation}\label{24}
-\Delta_y v_r = 1+\lambda v_r, \qquad v_r\in W^{1,2}_0(Y_{1,r}).
\end{equation}

If $0<r_1\le r_2$ and
the set 
$\{\nu_{j,r}: j\in \N,\, r \in [r_1,r_2]\}$
has gaps, then for 
$\lambda\in \R\setminus \bigl\{\nu_{j,r}: j\in \N,\, r \in [r_1,r_2]\bigr\}$
the solution to (\ref{23}) is given by (\ref{10}),
where the functions $v_{\widetilde \omega}(y)$ solve (\ref{24}) with $r= \widetilde \o_0$. Substituting it into the spectral problem for (\ref{16}) yields the problem (\ref{mac_eq})
with the Zhikov-type function $\beta$ given by ({\it cf.} (\ref{beta}))
\begin{equation}
\beta(\lambda):= \lambda \big(1 + \lambda \langle v \rangle_\O
\big) = \,\lambda\,+\, \lambda^2 \int\limits_{\{\widetilde{\o}_0\in[r_1,r_2]\}} \sum_{j=1}^\infty
\bigl(\nu_{j,\widetilde{\o}_0}-\lambda\bigr)^{-1}\biggl(\int_{Y_{1, \widetilde{\o}_0}}\varphi_{j, \widetilde{\o}_0}\biggr)^2
\,dP(\widetilde{\o}).
\label{beta_1}
\end{equation}
The integral in (\ref{beta_1}) is well defined for $\lambda\in \R\setminus
\{\nu_{j,r}: j\in \N,\, r \in [r_1,r_2]\},$ and the description of the spectrum on the intervals where $\beta(\l)>0$ follows Section \ref{simple_example}.

\begin{theorem}\label{thm:4.1} Under the assumptions 
of the current subsection, the spectrum of $\mathcal A$ is given by
	\begin{equation*}
	\sigma (\mathcal A) = \bigg(\bigcup_{j\in \N,\, r \in [r_1,r_2]} \bigl\{\nu_{j,r}, \nu_{j,r}'\bigr\}\Bigg) \cup 
	\bigl\{\l_{j,k}: j,k\in \N\bigr\}
	\end{equation*}
	where for each $k,$ the values $\l_{j,k}$ are solutions to $\beta(\l_{j,k}) = \mu_k,$ see (\ref{lambdas}). 
The point spectrum of the operator $\mathcal A$ is given by $\bigl\{\l_{j,k}: j,k\in \N\bigr\}.$
\end{theorem}
It is clear that if the set $\bigcup_{j\in \N,\, r \in [r_1,r_2]}\{\nu_{j,r}, \nu_{j,r}'\}$ has gaps, then $\sigma (\mathcal A)$ also has gaps. We are going to prove the theorem in several steps formulated in the following lemmas. We begin by studying the spectrum of the ``microscopic'' part of the limit operator.

\begin{lemma}\label{lem:4.2} The spectrum of the operator $-\Delta_\o$ (see (\ref{delta_omega_form})) is given by
	\begin{equation*}
	\sigma (-\Delta_\o) = \bigcup_{j\in \N,\, r \in [r_1,r_2]}\bigl\{\nu_{j,r}, \nu_{j,r}'\bigr\}
	\end{equation*}
	and does not contain eigenvalues of $-\Delta_\o.$
\end{lemma} 
\begin{proof}
	Let $\l =  \nu_{j_0,r_0}$  for some $j_0\in\N,r_0\in [r_1,r_2],$ and assume that $v\in W^{1,2}_0(\mathcal{O})$ is an eigenfunction corresponding to $\l,$ {\it i.e.}
	$-\Delta_{\omega} v = \l v.$ (For $\l =  \nu_{j_0,r_0}'$ argument is similar.) 
	Then $v$ is of the form (\ref{10}), where 
	$-\Delta v_{\widetilde \o} = \l v_{\widetilde \o}$ in 
	$Y_{1,\widetilde \o_0},$ 
	whenever $\widetilde \o_0\in [r_1,r_2].$ 
	But $\l$ is only an eigenvalue of the operator $-\Delta$ acting in $W^{1,2}_0(Y_{1,\widetilde \o_0})$ 
	if $\widetilde \o_0 = r_0$, hence
	\begin{equation}
	v(\o) = v(\widetilde \o,y) = \left\{ \begin{array}{ll}
	\varphi_{j_0,r_0}(y), & \widetilde \o_0 = r_0,\,y\in Y_{1,r_0},
	\\[0.4em]
	0 & \mbox{otherwise}.
	\end{array}\right.
	\end{equation}
	It remains to observe that $\{\widetilde \o_0 = r_0,\,y\in Y_{1,r_0}\}$ is a set of measure zero in $\Omega$ and hence $\|v\|_{L^2(\O)}=0$. The second claim of the lemma follows.
	
	Now we show that $\l\in\sigma (-\Delta_\o)$ by constructing  a Weyl sequence. Without loss of generality we can assume that $r_0\in (r_1,r_2)$.  For small enough $\delta>0$ we choose an $L^2$-function $w_\delta=w_\delta(r)$ such that $\supp w_\delta \subseteq (r_0-\delta,r_0+\delta)$ and $\|w_\delta\|_{L^2(r_0-\delta,r_0+\delta)}=1,$ {\it e.g.} we can choose $w_\delta$ to be equal to a constant proportional to $\delta^{-1/2}$ on  $(r_0-\delta,r_0+\delta)$. Consider the sequence     
	\begin{equation*}
	v_\delta(\o) = v_\delta(\widetilde \o,y) = \left\{ \begin{array}{ll}
	w_\delta(\widetilde \o_0)\varphi_{j_0,\widetilde \o_0}(y), & (\widetilde \o,y)\in \mathcal{O},
	\\[0.4em]
	0 & \mbox{otherwise}.
	\end{array}\right.
	\end{equation*}
	We have $v_\delta\in W^{1,2}_0(\mathcal{O})$, $\|v_\delta\|_{L^2(\Omega)}=1$  and 
	\begin{equation*}
	-\Delta_\o v_\delta(\o, y) = - w_\delta(\widetilde \o_0) \Delta \varphi_{j_0,\widetilde \o_0}(y) =\nu_{j_0,\widetilde \o_0} w_\delta(\widetilde \o_0)\varphi_{j_0,\widetilde \o_0}(y),
	\end{equation*}
	hence
	\begin{equation*}
	\bigl\|-\Delta_\o v_\delta-\l v_\delta\bigr\|^2_{L^2(\mathcal{O})} =
	\int\limits_{\{{\widetilde{\o}_0\in[r_1,r_2]\}}} \int\limits_{Y_{1,\widetilde{\o}_0}} \Big((\nu_{j_0,\widetilde \o_0} - \nu_{j_0,r_0})w_\delta(\widetilde \o_0)\varphi_{j_0,\widetilde \o_0}(y)\Big)^2 \,dy \,d\widetilde P(\widetilde{\o}) \to 0,\qquad \delta \to 0.
	\end{equation*}
	It follows that $v_\delta$ is a Weyl sequence for  $\l = \nu_{j_0,r_0}$.
	
	It remains to prove that $\l$ is in the resolvent set whenever $\l\notin\bigcup_{j\in \N,\, r \in [r_1,r_2]} \{\nu_{j,r}, \nu_{j,r}'\}$ . Assume the contrary and let $f\in L^2(\mathcal{O}),$ then the resolvent equation
	$-\Delta_\o v- \l v = f$
	has a unique solution given by (\ref{10}) with $v_{\widetilde \o}$ solving 
	$-\Delta v_{\widetilde \o}(y)- \l v_{\widetilde \o}(y) = f(\widetilde \o, y),\ $ $y\in Y_{1,\widetilde \o_0}.$
Moreover, since 
\[
d:={\rm dist}\bigg( \l,\ \bigcup_{j\in \N,\, r \in [r_1,r_2]}\bigl\{\nu_{j,r}, \nu_{j,r}'\bigr\}\bigg) >0,
\]
we have  
$\|v_{\widetilde \o}\|_{L^2(Y_{1,\widetilde \o})} \leq d^{-1} \|f(\widetilde \o, \cdot)\|_{L^2(Y_{1,\widetilde \o_0})},$
and it follows immediately that 
$\|v\|_{L^2(\mathcal{O})} \leq d^{-1} \|f\|_{L^2(\mathcal{O})},$
which concludes the proof.
\end{proof}

Next, we focus on the spectrum of $\mathcal A$.

\begin{lemma} 
\label{lem:4.4}
The inclusion $\sigma (-\Delta_\o)\subset  \sigma (\mathcal A)$  holds.
\end{lemma}

\begin{proof}
	The proof of the inclusion $\bigl\{\nu_{j,r}': j\in \N,\, r \in [r_1,r_2]\bigr\}\subset  \sigma (\mathcal A)$ repeats the related part of the proof of the Lemma \ref{lem:4.2}. Namely, for  $\l =  \nu_{j_0,r_0}'$, $j_0\in \N$, $r_0\in (r_1,r_2),$ we define a Weyl sequence $u_\delta := u_0^\delta + u_1^\delta \in V,$  where $u_0^\delta \equiv 0$ and  $u_1^\delta$ is given by
	$u_1^\delta(x,\o) : = f(x) v_\delta(\o),$ 
	with an arbitrary fixed $f\in L^2(S)$ and $v_\delta$ defined as in Lemma \ref{lem:4.2}. 

	In order to show that $\bigl\{\nu_{j,r}: j\in \N,\, r \in [r_1,r_2]\bigr\}\subset  \sigma (\mathcal A),$ suppose that $\l = \nu_{j_0,r_0}$ for some $j_0\in \N$ and $r_0\in (r_1,r_2)$. Assume, to the contrary, that there exists a bounded resolvent  $(\mathcal A -\l)^{-1},$ {\it i.e.} the system (\ref{mik1})--(\ref{mik2}) has a unique solution for all $f\in L^2(S\times \Omega)$. For $f=f(x)\in  L^2(S)$ the second equation reads
	\begin{equation*}
	\int\limits_{\{{\widetilde{\o}_0\in[r_1,r_2]\}}} \int\limits_{Y_{1,\widetilde{\o}_0}} (-\Delta_y u_1 - \l u_1)  \varphi_1 \,dy \,dP(\widetilde{\o}) = 
	 (f+\lambda u_0) \int\limits_{\{{\widetilde{\o}_0\in[r_1,r_2]\}}} \int\limits_{Y_{1,\widetilde{\o}_0}}   \varphi_1 \,dy \,d\widetilde P(\widetilde{\o}).
	\end{equation*}
	Then $u_1$ must necessary be of the form $u_1 = (f+\lambda u_0) v$, where $v$ is of the form (\ref{10}) and 
	$-\Delta_y v_{\widetilde \o} = \l v_{\widetilde \o} + 1$ in $Y_{1, \widetilde{\o}_0},$ {\it i.e.}
	\begin{equation}
	\label{55}
	v_{\widetilde\o} =  \sum_{j=1}^\infty
	(\nu_{j,\widetilde{\o}_0} - \l)^{-1}\biggl(\int_{Y_{1,\widetilde{\o}_0}}\varphi_{j,\widetilde{\o}_0}\biggr)
	\varphi_{j,\widetilde{\o}_0}.
	\end{equation}
	which clearly blows up as $\widetilde{\o}_0\to r_0.$
	We show that the corresponding $v$ is not an element of $L^2(\mathcal{O}),$ leading to a contradiction. Indeed, using the identity
	\begin{equation*}
	\nu_{j_0,\widetilde \o_0} - \nu_{j_0,r_0} = \widetilde \o_0^{-2} \nu_{j_0} - r_0^{-2} \nu_{j_0} = r_0^{-2}\widetilde \o_0^{-2}(r_0 - \widetilde \o_0)(r_0 + \widetilde \o_0)\nu_{j_0},
	\end{equation*} 
	one has
	\begin{align*}
	&\int\limits_{\{{\widetilde{\o}_0\in[r_1,r_2]\}}} \int\limits_{Y_{1,\widetilde{\o}_0}} |v|^2 \,dy \,dP(\widetilde{\o}) =  \int\limits_{\{{\widetilde{\o}_0\in[r_1,r_2]\}}} \sum_{j=1}^\infty
	\bigl(\nu_{j,\widetilde \o_0} - \nu_{j_0,r_0}\bigr)^{-2}
	\biggl(\int_{Y_{1, \widetilde{\o}_0}}\varphi_{j,\widetilde \o_0}\biggr)^2
	dP(\widetilde{\o})
	\\
	&\geq\int\limits_{\{{\widetilde{\o}_0\in[r_1,r_2]\}}}
	\bigl(\nu_{j,\widetilde \o_0} - \nu_{j_0,r_0}\bigr)^{-2}
	\biggl(\int_{Y_{1, \widetilde{\o}_0}}\varphi_{j,\widetilde \o_0}\biggr)^2
	dP(\widetilde{\o})\geq	C \int\limits_{\{{\widetilde{\o}_0\in[r_1,r_2]\}}}
	|\widetilde \o_0 - r_0|^{-2}
	\,dP(\widetilde{\o}),
	\end{align*}
	where the last integral diverges. 
\end{proof}

\begin{lemma}If $\beta(\lambda) = \mu_k$ for some $k\in \N$ then $\lambda$ is an eigenvalue of $\mathcal A.$ 
\end{lemma}
\begin{proof}
Let $\beta(\lambda) = \mu_k$, and denote by $\psi_k \in W^{1,2}_0(S)$ be the corresponding eigenfunction of $-\div A^{\rm hom}_1 \nabla$. Since $\l\notin 
\{\nu_{j,r}: j,r\in{\mathbf N}\},$ the problem 
		\begin{equation*}
v\in W^{1,2}_0(\mathcal{O}),	\qquad -\Delta_\o v = \l v + 1, 
	\end{equation*}
	has a solution given by (\ref{10}), (\ref{55}). Therefore  $u= \psi_k + \l \psi_k v$ is the eigenfunction of $\mathcal A$ corresponding to $\l$.
\end{proof}

\begin{lemma}A point $\lambda$ belongs to the resolvent set of $\mathcal A$ if 
	$\lambda \notin \sigma (-\Delta_\o)$
	and $\beta(\lambda)\notin \sigma (-\div A_1^{\rm hom}\nabla)$, {\it i.e.} $\beta(\lambda)<0$ or $\beta(\lambda)\geq0$ and $\beta(\lambda)\neq \mu_k$, $k\in \N$.
\end{lemma}
\begin{proof}
  We claim that the problem (\ref{mik1})--(\ref{mik2}) has bounded resolvent. Indeed, suppose that $f\in L^2(S\times \Omega)$ and write (\ref{mik2}) in the form
	\begin{equation*}\label{36}
-\Delta_\o u_1  - \lambda u_1  =  \lambda u_0 + f.
	\end{equation*}
	Since $\l$ is not in the spectrum of $-\Delta_\o$, the latter has a bounded resolvent at $\l$ and $u_1=\lambda u_0 v+g$, where $v = (-\Delta_\o - \lambda)^{-1} 1$ is as in (\ref{10}) and $g=g(x,\o) = (-\Delta_\o - \lambda)^{-1}  f(x,\o),$ $x\in S.$ In particular,  
	\begin{equation*}
	\|v\|_{L^2(\mathcal{O})} \leq {\rm dist}\bigl(\l,\sigma(-\Delta_\o)\bigr)^{-1}\|1\|_{L^2(\mathcal{O})},\qquad
	\bigl\|g(x,\cdot)\bigr\|_{L^2(\mathcal{O})} \leq {\rm dist}\bigl(\l,\sigma(-\Delta_\o)\bigr)^{-1} \bigl\|f(x,\cdot)\bigr\|_{L^2(\Omega)}.
	\end{equation*}
	 Substituting the expression for $u_1$ in (\ref{mik1}) we obtain  
	\begin{equation*}
	-\div A^{\rm hom}_1 \nabla u_0 - \beta(\lambda) u_0 =\langle f + \lambda g \rangle.
	\end{equation*}
	For  $\beta(\lambda)\notin \sigma (-\div A_1^{\rm hom}\nabla)$ the operator $-\div A^{\rm hom}_1 \nabla - \beta(\lambda)$ is invertible and  
		\begin{align*}
	\|u_0\|_{L^2(S)}& \leq {\rm dist}\bigl(\beta(\l),\sigma\bigl(-\div A^{\rm hom}_1 \nabla\bigr)\bigr)^{-1}\bigl\|\langle f + \lambda g \rangle\bigr\|_{L^2(S)} \\[0.4em]
	&\leq {\rm dist}\bigl(\beta(\l),\sigma\bigl(-\div A^{\rm hom}_1 \nabla\bigr)\bigr)^{-1} \big(\|f\|_{L^2(S\times \Omega)} + |\lambda| \|g \|_{L^2(S\times\mathcal{O})}\big),
	\end{align*}
	from which the claim follows.
\end{proof}


\begin{proposition}The set 
$ \sigma (-\Delta_\o) \setminus 
 \bigl\{\l_{j,k}: j,k\in \N\bigr\}$
does not contain eigenvalues of the operator $\mathcal A$.
\end{proposition}
\begin{proof}
	Assume that $\l = \nu_{j_0,r_0}$, for some $j_0\in \N$ and $r_0\in [r_1,r_2]$, is an eigenvalue of  $\mathcal A$, {\it i.e.} there exists $u=u_0+u_1 \in V$ such that 
		\begin{align}
	&-\div A^{\rm hom}_1 \nabla u_0 = \lambda\bigl(u_0 + \langle u_1 \rangle_\mathcal{O}\bigr),\qquad\qquad\label{64Igor}
	\\[0.4em]
	&-\Delta_\o u_1(x,\cdot) = \lambda\bigl(u_0(x) +  u_1(x,\cdot)\bigr).\nonumber
	\end{align}
	Suppose that $u_0(x)\neq 0$ for some $x\in S,$ then $u_1(x,\cdot) = \l u_0(x)v(\cdot)$, where $v$ solves
	\begin{equation}
	\label{64}
	-\Delta_\o v = \l v + 1.
	\end{equation}
	Arguing as for the second inclusion of Lemma \ref{lem:4.4}, we see that (\ref{64}) has no $L^2$-solution for the given $\l$. It follows that $u_0 = 0$ and therefore $u_1(x,\cdot)$ is an eigenfunction of $-\Delta_\o,$ which cannot be true by Lemma \ref{lem:4.2}. 
	
	Now we assume that $\l = \nu_{j_0,r_0}'$. Arguing as above, for $u_0(x)\neq 0$ we have $u_1(x,\cdot) = \l u_0(x)v(\cdot)$, where $v$ solves (\ref{64}). The solution exists and is given by (\ref{10}), (\ref{55}). Substituting $u_1$ into (\ref{64Igor}) we see that $u_0$ must satisfy 
	$-\div A^{\rm hom}_1 \nabla u_0 = \beta(\lambda) u_0,$
	which cannot be true since $\beta(\l)\notin\sigma\bigl(-\div A^{\rm hom}_1 \nabla\bigr).$
	
	Finally, if $u_0 = 0,$ 
	then we argue as above for the case $\l = \nu_{j_0,r_0},$ again arriving at a contradiction.
\end{proof}
This completes the proof of Theorem \ref{thm:4.1}.


\section*{Acknowledgments}
KC is grateful for the support of
the Engineering and Physical Sciences Research Council: Grant EP/L018802/2 ``Mathematical foundations of metamaterials: homogenisation, dissipation and operator theory''. IV has been supported by the Croatian Science
Foundation under Grant agreement No.~9477 (MAMPITCoStruFl).


\begin{thebibliography}{DMM86b}
			\bibitem[AK81]{Ackoglu} M. A. Ackoglu, U. Krengel. 
			\newblock
			Ergodic theorems for superadditive processes.   
			\newblock{\em J. Reine Angew. Math.} {\bf 323}, 53--67 (1981). 
		\bibitem[All92]{Allaire}
		G. Allaire.
		\newblock Homogenization and two-scale convergence.
		\newblock {\em SIAM J. Math. Anal.}, {\bf 23}(6), 1482--1518 (1992).
       \bibitem[Bel17]{Bellieud}
	M. Bellieud.
                 \newblock Homogenization of stratified elastic composites with high contrast.
            \newblock{\em SIAM Journal Math. Anal.}, {\bf 49}(4), 2615--2665 (2017).
	\bibitem[BBM15]{BBM}
	G. Bouchitt\'{e}, C. Bourel, and L. Manca.
	\newblock Resonant effects in random dielectric structures.
	\newblock{\em ESAIM: Control, Optimisation and Calculus of Variations}, {\bf 21}, 217--246 (2015).
	\bibitem[BMP03]{BMP}
	A. Bourgeat, A. Mikeli{\'c}, and A. Piatnitski.
	\newblock On the double porosity model of a single phase flow in random media.
	\newblock {\em Asymptotic Analysis}, {\bf 34}, 311--332 (2003).
	\bibitem[BMW94]{mikelic}
	A. Bourgeat, A. Mikeli{\'c}, and S. Wright.
	\newblock Stochastic two-scale convergence in the mean and applications.
	\newblock {\em J. Reine Angew. Math.}, {\bf 456}, 19--51 (1994).	
	
	\bibitem[CC16]{ChC} K. Cherednichenko and S. Cooper.
	\newblock Resolvent estimates for high-contrast homogenisation problems. 
	\newblock {\em Archive for Rational Mechanics and Analysis} {\bf 219}(3), 1061--1086 (2016).
	
	\bibitem[CCC17]{ChChC} K. Cherednichenko, M. Cherdantsev, and S. Cooper.
	\newblock Extreme localisation of eigenfunctions to one-dimensional high-contrast periodic problems with a defect.
	\newblock {\em arXiv:1702.03538} (2017).
	
	
	\bibitem[CFS82]{cornfeld}
        I. P. Cornfeld, S. V. Fomin, and Y. G. Sinai.
	\newblock {\em Ergodic Theory.}
	\newblock Grundlehren der Mathematischen Wissenschaften Vol $245$,
      Springer-Verlag, New York (1982).
	\newblock Translated from the Russian by A. B. Sossinskii.
	\bibitem[DMM86a]{dalmasomodica1}
	G. Dal~Maso and L. Modica.
	\newblock Nonlinear stochastic homogenization.
	\newblock {\em Ann. Mat. Pura Appl. (4)}, {\bf 144}, 347--389 (1986).	
	\bibitem[DMM86b]{dalmasomodica2}
	G. Dal~Maso and L. Modica.
	\newblock Nonlinear stochastic homogenization and ergodic theory.
	\newblock {\em J. Reine Angew. Math.}, {\bf 368}, 28--42 (1986).	
	\bibitem[DG16]{gloria1}
	M. Duerinckx and A. Gloria.
	\newblock Stochastic homogenization of nonconvex unbounded integral functionals
	with convex growth.
	\newblock {\em Arch. Ration. Mech. Anal.}, {\bf 221}(3), 1511--1584 (2016).
		\bibitem[RS80]{Simon1}
	M. Reed and B. Simon.
	\newblock {\em Methods of Modern Mathematical Physics. I: Functional Analysis.}
	\newblock Academic Press,
	 New York (1980).
	\bibitem[SW11]{sango1}
	M. Sango and J.-L. Woukeng.
	\newblock Stochastic two-scale convergence of an integral functional.
	\newblock {\em Asymptot. Anal.}, {\bf 73}(1-2), 97--123 (2011).    
    \bibitem[Sh96]{Shiryaev1996}
    A. N. Shiryaev. 
    \newblock {\em Probability.}
    \newblock Springer-Verlag (1996).
   
    \bibitem[ZKO94]{zhikov2}
    V. V. Zhikov, S. M. Kozlov, and O. A. Ole{\u\i}nik.
    \newblock  {\em Homogenization of Differential Operators and Integral
    Functionals.}
    \newblock Springer-Verlag, Berlin (1994).
    
    \bibitem[Zh05]{zhikov2005}	V. V. Zhikov, Gaps in the spectrum of some elliptic operators in divergent form with periodic coefficients. {\itshape (Russian) Algebra i Analiz} 16, no. 5, 34--58; translation in St. Petersburg Math. J. {\bf 16}(5), 773--790 (2004). 
    
    
    \bibitem[Zh00]{Zhikov2000}
    V. V. Zhikov, On an extension of the method of two-scale
    convergence and its applications. {\itshape (Russian) Mat. Sb.} {\bf 191}(7), 31--72; {\itshape translation in Sb. Math.,}
    {\bf 191}(7--8), 973--1014 (2000). 
    
   
	\bibitem[ZP06]{zhikov1}
	V. V. Zhikov and A. L. Pyatnitski{\u\i}.
	\newblock Homogenization of random singular structures and random measures.
	\newblock {\em Izv. Ross. Akad. Nauk Ser. Mat.}, {\bf 70}(1), 23--74 (2006).
\end{thebibliography}
\end{document}